\documentclass[onecolumn]{article}

\usepackage[utf8]{inputenc}
\usepackage[T1]{fontenc}
\usepackage{amsmath}
\usepackage{amssymb}
\usepackage{amsthm}
\usepackage{mathtools}
\usepackage{graphicx}
\usepackage{rotating}     
\usepackage{subcaption}    
\usepackage{enumitem}
\usepackage{booktabs}
\usepackage{multirow}
\usepackage[labelfont=bf]{caption}
\usepackage{geometry}
\usepackage{ragged2e}
\usepackage{microtype}
\usepackage{hyperref}
\usepackage{cleveref}
\usepackage{physics}
\usepackage{tabularx}
\usepackage{float}
\usepackage{algorithm}
\usepackage{algpseudocode}
\usepackage{amsmath,amssymb}
\usepackage{mathtools}

\usepackage{braket}

\def\BibTeX{{\rm B\kern-.05em{\sc i\kern-.025em b}\kern-.08em
    T\kern-.1667em\lower.7ex\hbox{E}\kern-.125emX}}
\usepackage[sorting=none]{biblatex}
\usepackage{hyperref}
\hypersetup{
    colorlinks=true,
    linkcolor=blue,
    filecolor=magenta,
    urlcolor=cyan,
    pdftitle={Your Title Here},
    pdfauthor={Your Name Here},
}

\newtheorem{lemma}{Lemma}
\newtheorem{corollary}{Corollary}[section]

\title{\bfseries A Polylogarithmic-Time Quantum Algorithm for the Laplace Transform}

\author{
    \small{Akash Kumar Singh\textsuperscript{1,2}} \and
    \small{Ashish Kumar Patra\textsuperscript{2}} \and
    \small{Anurag K.S.V\textsuperscript{2}} \and
    \small{Sai Shankar P.\textsuperscript{2}} \and
    \small{Ruchika Bhat\textsuperscript{3}} \and
    \small{Jaiganesh G.\textsuperscript{2,}}\thanks{Corresponding Author: drjaiganesh15@gmail.com, jaiganesh@qclairvoyance.in} \and
}

\date{}

\begin{document}

\maketitle
\vspace{-1.5em}
\begin{center}
\small{\textsuperscript{1}\,School of Quantum Technology, Defence Institute of Advanced Technology, Pune, MH 411025, India} \\[3pt]
\small{\textsuperscript{2}\,Qclairvoyance Quantum Labs, Secunderabad, TG 500094, India} \\[3pt]
\small{\textsuperscript{3}\,The University of Arizona, Tucson, AZ 85721, USA} \\[6pt]
\end{center}
\vspace{2em}

\begin{abstract}
We introduce a quantum algorithm to perform the Laplace transform on quantum computers. Already, the quantum Fourier transform (QFT) is the cornerstone of many quantum algorithms, but the Laplace transform or its discrete version has not seen any efficient implementation on quantum computers due to its dissipative nature and hence non-unitary dynamics. However, a recent work has shown an efficient implementation for certain cases on quantum computers using the Taylor series. Unlike previous work , our work provides a completely different algorithm for doing Laplace Transform using Quantum Eigenvalue Transformation and Lap-LCHS, very efficiently at points which form an arithmetic progression. Our algorithm can implement $N \times N$ discrete Laplace transform in gate complexity that grows as $\order{(\log N)^{3}}$, assuming the efficient state preparation, where $N=2^n$ and $n$ is the number of qubits,
which is a superpolynomial speedup in number of gates over the best classical counterpart that has complexity $\order{N \log N}$ for the same cases. Also, the circuit width grows as $\order{\log N}$. Quantum Laplace Transform (QLT) may enable new Quantum algorithms for cases like solving differential equations in the Laplace domain, developing an inverse Laplace transform algorithm on quantum computers, imaginary time evolution in the resolvent domain for calculating ground state energy, and spectral estimation of non-Hermitian matrices.

\end{abstract}
\vspace{1em}

\section{Introduction}
\label{sec:introduction}

The Laplace transform stands as an indispensable mathematical tool across diverse fields of physics and engineering. Its utility spans from simplifying ordinary and partial differential equations into more manageable algebraic forms to extensive applications in signal processing. Classically, numerical implementations of the discrete Laplace transform for an $N \times N$ matrix typically incur a computational scaling of $\order{N \log N}$~\cite{Rabiner1969} in some special cases, or degrade to $\order{N^2}$~\cite{Loh2023FastDLT} in worst-case scenarios. In terms of precision it scales as $\order{\log(1/\epsilon)}$, where $\epsilon$ is the desired accuracy. While direct Laplace transform computation might not always be the bottleneck in classical algorithms, the inverse Laplace transform often poses significant numerical challenges when solving differential equations.

In contrast, quantum computing has demonstrated remarkable speedups for certain computational problems. The Quantum Fourier Transform (QFT)~\cite{coppersmith2002_qft}, for instance, provides a significant acceleration for the Discrete Fourier Transform (DFT), reducing its gate complexity from classical $\order{N \log N}$ to quantum $\order{(\log N)^2}$. The QFT forms the bedrock of numerous influential quantum algorithms, including Quantum Phase Estimation (QPE)~\cite{kitaev1995quantummeasurementsabelianstabilizer} and Shor's algorithm. A similar breakthrough for a QLT could revolutionize the field of quantum computing.

However, the efficient implementation of a QLT has remained an inherently difficult and largely unresolved problem. The primary challenge stems from the dissipative nature of the Laplace transform, which does not preserve the norm of quantum states. Consequently, the corresponding operator is non-unitary, a property that is incompatible with the unitary evolution fundamental to quantum computation. This gap in quantum algorithm development, specifically for other commonly used transforms beyond the Fourier transform, motivated our investigation.

\subsection{Related works}
Recently, Zylberman et. al.\cite{zylberman2024fastlaplacetransformsquantum} gave an algorithm to implement the QLT, which
achieves polynomial scaling or better for both gate complexity and width of quantum circuit in number of qubits and
precision for some specific cases. More specifically, whenever the non-unitary diagonal operator in their
algorithm is implementable in polynomial depth in qubits and precision, their algorithm can also be
implemented in polynomial depth. Also, for some cases when a diagonal operator can be implemented
with circuit depth independent of qubits `n', their QLT algorithm can achieve a double exponential speed
up over classical methods with depth scaling as $\order{log(logN)}$ where `N' is $2^n$. The problem with
the algorithm is that for an arbitrary n-qubit non-unitary diagonal operator, it generally scales 
exponentially with `n' in terms of depth or circuit size.

In this work we utilize Linear Combination of Hamiltonian Simulation (LCHS). This framework is conceptually distinct from the algorithmic methodology of ~\cite{zylberman2024fastlaplacetransformsquantum}.

In this work we utilize Linear Combination of Hamiltonian Simulation (LCHS). This framework is conceptually distinct from the algorithmic methodology of ~\cite{zylberman2024fastlaplacetransformsquantum}.LCHS basically does a Quantum Eigenvalue Transformation. Quantum Algorithms for QSVT(Quantum Singular Value Transformation) are well known, but no single technique or strategy is known for quantum eigenvalue transformation, especially for non-normal matrices. Recently, Dong et al. \cite{an2024laplacetransformbasedquantum} provided a way to do a specific type of matrix Laplace transform, which can be interpreted as performing a Laplace transform on eigenvalues of a particular matrix using LCHS. 

We build upon Lap-LCHS framework to perform the Laplace transform very efficiently on an actual vector, which is not immediately evident from its original formulation. By ‘efficiently’, we mean that the transform can be implemented with polylogarithmic gate complexity and a logarithmic asymptotic scaling for width of the circuit in the terms of $N$ which is the number of rows or columns in the Discrete Laplace Transform matrix ($N=2^n$, where n is the number of qubits) by exploiting the arithmetic-progression structure of the Laplace variable $s$.

\subsection{Our contributions}
The QLT is a quantum algorithmic implementation of the Laplace Transform by discretizing it, drawing a parallel to how Quantum Fourier Transform (QFT) enables the efficient computation of the Discrete Fourier Transform (DFT).\\
The classical Laplace transform of a function $g(t)$ is defined as
\begin{equation}
\mathcal{L}\{f(t)\} = \int_{0}^{\infty} e^{-st} g(t) dt
\label{eq:classical_laplace}
\end{equation}

The discretized Laplace transform matrix acting on a vector of length 
$N = 2^n$ can be defined by discretizing time $t$ and the Laplace variable $s$. 
By introducing an appropriate normalization factor, this matrix naturally leads 
to the definition of the \textit{QLT}, which we also denote as 
$\hat{\mathrm{QLT}}$. Thus, the $\hat{QLT}$ is obtained as the normalized version of the 
discretized Laplace transform matrix. This matrix is non-unitary and can be represented as:
\begin{equation}
\hat{QLT} = \frac{1}{N} \begin{pmatrix}
e^{-s_1 t_1} & \dots & e^{-s_1 t_N} \\
\vdots & \ddots & \vdots \\
e^{-s_N t_1} & \dots & e^{-s_N t_N}
\end{pmatrix}
\label{eq:qlt_matrix}
\end{equation}

The exact correspondence to QFT can be established if we omit the negative exponent sign; the parameter $t$ can be discretized to range $0$ to $N$ with unit increments, and the Laplace parameter $s$ can adopt complex integer forms with constant differences like $1+1i$, $2+2i$...\\
Our contributions primarily involve leveraging Quantum Eigenvalue Transformation (QET) through Lap-LCHS \cite{an2024laplacetransformbasedquantum} to
efficiently implement the QLT.
The quantum algorithm described here gives an efficient implementation for the constant difference case and even when the real and imaginary parts form separate arithmetic progressions, given that the convergence condition $Re(s) \ge 0$ and $g(t) \in L^1(\mathbb{R}_{+})$ are met. Although our proof focuses on the case where $a_j=b_j$ and form the same arithmetic progression, extending the argument to $a_j \ne b_j$ with independent arithmetic progressions is straightforward and maintains the identical asymptotic speed-up. There is a possibility of getting the same speedup by using real and imaginary part values sampled from certain polynomials, proving which is out of scope of this work. Here $a_j$ and $b_j$ are the real and imaginary parts of the Laplace variable $'s'$.
Beyond these insight and proofs, our key contributions also are:

\begin{itemize}

\item \textbf{Encoding the Laplace Variable}: 
 We encode the Laplace variable $'s'$ values into the eigenvalues $\lambda_1, \lambda_2, \dots$ of a diagonal matrix $A$ . Then the transform $\mathcal{}h(A)$ has eigenvalues $\mathcal{}h(\lambda_1), \mathcal{}h(\lambda_2), \dots$. This is exactly the form of laplace transform evaluated at $s=\lambda_1,\lambda_2,\lambda_3...$.  Using this we can perform the QLT for any function $g(t) \in L^1(\mathbb{R}^+)$. This insight, not explicitly highlighted in prior work, is a cornerstone of our algorithm.

\begin{equation}
A = \begin{pmatrix}
s_0 & 0 & 0 & \dots \\
0 & s_1 & 0 & \dots \\
0 & 0 & s_2 & \dots \\
\vdots & \vdots & \vdots & \ddots
\end{pmatrix}, \quad \text{where } s_j = a_j + i b_j
\end{equation}

\item \textbf{Encoding the information about $g(t)$}: The QLT matrix in \cite{zylberman2024fastlaplacetransformsquantum} does not contain any information about the function on which the Laplace transform is to be performed. The information about $g(t)$ needs to be encoded into a quantum state and then acted upon by the QLT matrix in block-encoded form \cite{zylberman2024fastlaplacetransformsquantum}. But in our algorithm, where we use Lap-LCHS, the information of the $g(t)$ is included in the unitary matrices which are part of both preparation and unpreparation parts of LCU circuit. The information about the Laplace variable, that is $s$, is in the $A$ matrix, and the information about $g(t)$ is encoded into a quantum oracle $U_l$ and its transpose, which prepares the corresponding encoded quantum state and then unprepares it.\textbf{ We can extract information if we act $h(A)$ (which is an overall unitary operator combining preparation, selection, and unpreparation) upon an equal superposition state.} We can then interpret it as if QLT matrix is subsequently applied to the encoded state $g(t)$ in block-diagonal form.

\item \textbf{Efficient SELECT Operator Implementation}:
We exploit the arithmetic progression structure and commutativity of matrices \( L \) and \( H \) (Hermitian and Anti-Hermitian part of $A$ resp.) to significantly reduce complexity of implementing $SELECT$ operator. We have done this by converting our discretized summation into a product of unitaries, which reduces the number of unitaries to be implemented from $ \mathcal {O}(N^2)$ to $\mathcal{O}\!\bigl(\log^{2} N\bigr)$. Each such unitary is implemented efficiently using basic elementary gates via single-step \textit{Trotterization}/Product formula. Given \( [H,L] = 0 \) (since both are diagonal matrices) and their entries forming an arithmetic progression in real and imaginary parts, a single-step Trotterization suffices which results in decomposition into controlled Phase gates and controlled Phase shift gates.

\item \textbf{Reduced Complexity via Controlled Gates Optimization}: 
We eliminated multi-controlled unitary requiring more than two control qubits.\textbf{ Our optimized SELECT operator only needs unitaries that are controlled on at most two qubits at a time.} The number of such controlled unitaries scales as \(O((\log N)^2)\), and after Trotterization, the number of controlled $R_Z$ and  Phaseshift gates required in total scales as $(2 + 3\log N)$ for each of such controlled unitaries. Thus, the total gate complexity of our quantum algorithm with asymptotic scaling is \(\mathcal{O}((\log N)^3)\), compared to the best classical complexity of \(\mathcal{O}(N \log N)\). Further optimization is possible; although this does not alter the asymptotic scaling, this can deliver meaningful practical speed-ups. This refinement will be discussed later.

\item \textbf{Oracle Simplification}:
We have done \textbf{further optimization by removing the additional registers} \(|k_j\rangle\) and \(|t_l\rangle\) along with associated oracles \(O_k\) and \(O_t\) of Lap-LCHS, previously necessary \cite{an2024laplacetransformbasedquantum}. Our current formulation no longer requires these additional oracles. We have also removed other unnecessary registers and oracles that were needed to create the block encoding for time evolution using QSVT. Since we are not using QSVT but trotterization, we don't need those extra registers and oracles.

\end{itemize}

Collectively, these improvements establish a QLT algorithm that achieves an superpolynomial speedup by exploiting the arithmetic-progression structure of the Laplace variable. Our approach reduces the gate complexity to $\mathcal{O}((\log N)^3)$ and surpass the classical $\mathcal{O}(N \log N)$ limit.

The remainder of this paper is organized as follows. Section 2 reviews the necessary preliminaries, including block encoding, the LCU framework, Trotterization, and the Lap-LCHS method. Section 3 presents our quantum algorithm for the Laplace transform, detailing the construction of the PREP, SELECT, and UNPREP operators. Section 4 provides the complexity analysis, including circuit width, gate scaling, and proofs of the key structural lemmas. Section 5 describes our implementation strategy and simulation on Qiskit and Pennylane. Finally, Section 6 concludes with a discussion of limitations, potential applications, and future research directions.

\section{Preliminaries}
\label{sec:preliminaries}
This section introduces the quantum algorithmic primitives necessary for understanding our proposed QLT algorithm. These include block encoding \cite{shantanav_block_encoding}, Linear Combination of Unitaries (LCU) \cite{Childs_2017_lcu}, Trotterization \cite{lloyd1996universal}, and Laplace-based Linear Combination of Hamiltonian Simulation (Lap-LCHS) \cite{an2024laplacetransformbasedquantum}.

\subsection{Block encoding}

Block encoding \cite{lin2022lecturenotesquantumalgorithms} is a technique which embeds a non-unitary matrix A into the top-left corner of a higher-dimensional unitary matrix U. The unitary U is defined as an $(\alpha, a, \epsilon)$-block encoding of A if:

\begin{equation}
    U = \begin{pmatrix}
        \frac{A}{\alpha} & * \\ * & *
    \end{pmatrix}
\end{equation}
and $||A-\alpha(\bra{0^a}U|\ket{0^a})||\leq \epsilon$, $\alpha$ is the normalization factor and $a$ is the number of ancilla qubits.
\subsection{Linear Combination of Unitaries (LCU)}
The Linear Combination of Unitaries (LCU) technique \cite{Childs_2017_lcu} provides a method to approximate a general matrix A, expressed as $A = \sum_j \alpha_j U_j$, by combining a set of unitary operations $U_j$ with corresponding weights $\alpha_j$ using control unitaries and ancilla post-selection. The LCU algorithm involves three main steps:
\begin{enumerate}
    \item \textbf{Prepare:} An ancilla state is prepared in a superposition proportional to the square roots of the normalized weights: $|\psi\rangle = \frac{1}{\lambda}\sum_j \sqrt{\alpha_j} |j\rangle$, where $\lambda$ is the normalization factor.
    \item \textbf{Select:} A controlled unitary operation $\sum_j |j\rangle \langle j| \otimes U_j$ is applied.
    \item \textbf{Uncompute/Unprep:} The ancilla state is uncomputed, and a measurement is performed to project onto the desired state.
\end{enumerate}
When A is implemented via block encoding and LCU, amplitude amplification or oblivious amplitude amplification can be employed to boost the success probability to near-deterministic levels \cite{childs2012hamiltonian, Gily_n_2019}. LCU is particularly effective when used in conjunction with block encoding to represent complex matrix functions like $A^{-1}$ or $e^{-A}$. A LCU schematic is shown in Figure \ref{fig:LCU}.

    \begin{figure}[htbp]
  \centering
  \includegraphics[width=0.7\textwidth]{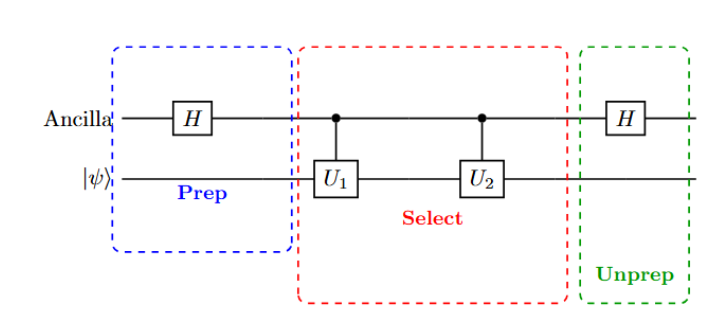}
  \caption{Linear Combination of Unitaries (LCU) circuit architecture acting on $\ket{\psi}$. The circuit is divided into three conceptual stages. In the \textit{Preparation} stage, a Hadamard gate prepares the ancilla register in a coherent superposition that encodes the weighting coefficients of the linear combination. In the \textit{Selection} stage, controlled unitaries $U_1, U_2, \ldots$ act on the system register conditioned on the ancilla state, implementing the weighted sum of operators through coherent control. In the \textit{Unpreparation} stage, is the uncomputation step. Then, a measurement is done to post-select the desired outcome.}

  \label{fig:LCU}
\end{figure}

\subsection{Trotterization}
Trotterization \cite{lin2022lecturenotesquantumalgorithms}, or the Trotter-Suzuki decomposition, is a fundamental technique in quantum simulation used to approximate the time-evolution operator $e^{-iHt}$ for a Hamiltonian H, especially when H can be decomposed into a sum of simpler parts, $H = \sum_j H_j$. This method is particularly useful when individual $H_j$ terms can be simulated efficiently, but the full Hamiltonian H cannot.

The first-order Trotter formula provides a basic approximation:
\begin{equation}
e^{-i(H_1+H_2)t} \approx \left(e^{-iH_1t/r}e^{-iH_2t/r}\right)^r
\label{eq:trotter_first_order}
\end{equation}
where $r$ denotes the number of Trotter steps. Higher-order formulas, such as the second-order symmetric Trotter-Suzuki formula, can improve accuracy with a moderate increase in cost. The approximation error generally scales as $\order{t^2/r}$ or better, depending on the commutativity of the individual Hamiltonian components, i.e., $[H_i, H_j]$. Trotterization is generally not favored for current NISQ hardware due to its high circuit complexity and depth \cite{lloyd1996universal, lin2022lecturenotesquantumalgorithms}. In our algorithm, due to the arithmetic progression structure of `s' and the commuting nature of our Hamiltonian decomposition, a single Trotter step is sufficient.

\subsection{Laplace-Based Linear Combination of Hamiltonian Simulation (Lap-LCHS)} \label{Lap-LCHS}

Lap-LCHS \cite{an2024laplacetransformbasedquantum} is the basis on which our algorithm is designed, so it is necessary to understand this completely before proceeding to understand our algorithm. Let \(A = L + iH\) be the Cartesian decomposition of the matrix $A$, with
\begin{equation}\label{eqn: lh_mats}
L = \frac{A + A^\dagger}{2}, 
\qquad
H = \frac{A - A^\dagger}{2i},
\end{equation}
and assume \(L \succeq 0\).  We wish to implement a target matrix function
\begin{equation}
h(A)
\;=\;
\int_{0}^{\infty} g(t)\,e^{-A t}\,dt,
\end{equation}
where \(g(t)\) is the inverse Laplace transform of \(h\).  By expressing
matrix exponential via a one‐dimensional LCHS integral with kernel \(f(k)\), we obtain the double‐integral representation
\begin{equation}
h(A)
\;=\;
\int_{0}^{\infty}\!\int_{-\infty}^{\infty}
f(k)\,\frac{g(t)}{1 - i k}
\;e^{-i\,t (\,k L + H\,)}\;dk\,dt.
\end{equation}
A convenient choice for \(f(k)\) is
\begin{equation}
f(k) = \frac{1}{2\pi e^{-2^\beta} e^{(1 + i k)^\beta}}, \quad \beta \in (0, 1)
\end{equation}
this function decays exponentially \cite{an2024laplacetransformbasedquantum}.

\textbf{Discretization.}
Truncate the integrals to \(k\in[-K,K]\) and \(t\in[0,T]\).  Introduce step
sizes \(h_k = 2K/M_k\), \(h_t = T/M_t\) and grid points
\begin{equation}
k_j = -K + j\,h_k,\quad j=0,\dots,M_k-1,
\qquad
t_\ell = \ell\,h_t,\quad \ell=0,\dots,M_t-1.
\end{equation}
Then approximate
\begin{equation}
h(A)
\approx
\sum_{j=0}^{M_k-1}\sum_{\ell=0}^{M_t-1}
\underbrace{\bigl(h_k\,f(k_j)/(1 - i k_j)\bigr)}_{c_j}
\;\underbrace{\bigl(h_t\,g(t_\ell)\bigr)}_{\hat c_\ell}
\;e^{\bigl(-i\,t_\ell\,(k_j L + H)\bigr)}.
\end{equation}

\textbf{Oracles.}
The algorithm requires:
\begin{itemize}
  
  \item Coefficient‐preparation unitaries
 \begin{equation}
O_{c,r}:\; \lvert0\rangle \mapsto \sum_{j=0}^{M_k-1} \sqrt{\frac{c_j}{\|c\|_1}}\,\lvert j\rangle, 
\qquad
O_{\hat c,r}:\; \lvert0\rangle \mapsto \sum_{\ell=0}^{M_t-1} \sqrt{\frac{\hat c_\ell}{\|\hat c\|_1}}\,\lvert \ell\rangle.
\end{equation}

  \item We also need the transpose of these coefficient oracles once they are represented as matrices.
  \item Simple oracles \(O_k\), \(O_t\) that map index registers \(\lvert j\rangle\),
    \(\lvert \ell\rangle\) to the binary encodings of \(k_j\) and \(t_\ell\).
\end{itemize}

\subsubsection{Lap-LCHS Algorithm}
\begin{enumerate}
    \item \textbf{Initial State Preparation:}
    \begin{equation}
    \ket{j} \ket{l} \ket{0}_j \ket{0}_l \ket{0}_R \ket{0}_{R_k} \ket{0}_{R_t} \ket{0}_a \ket{\psi}
    \end{equation}

    \item \textbf{Apply oracles } $O_k$ \textbf{ and } $O_t$ \textbf{ to compute discretization nodes:}
    \begin{equation}
    \ket{j} \ket{l} \ket{k_j}_j \ket{t_l}_l \ket{0}_R \ket{0}_{R_k} \ket{0}_{R_t} \ket{0}_a \ket{\psi}
    \end{equation}

    \item \textbf{Apply controlled rotation } $c\text{-}R$:
    \begin{equation}
    \ket{k}\ket{0} \rightarrow \ket{k} \left( \frac{\sqrt{\alpha_L k}}{\sqrt{\alpha_L |k| + \alpha_H}} \ket{0} + \frac{\sqrt{\alpha_H}}{\sqrt{\alpha_L |k| + \alpha_H}} \ket{1} \right)
    \end{equation}

    Resulting in:
    \begin{equation}
    \ket{j} \ket{l} \ket{k_j}_j \ket{t_l}_l \left( \frac{\sqrt{\alpha_L k_j}}{\sqrt{\alpha_L |k_j| + \alpha_H}} \ket{0}_R + \frac{\sqrt{\alpha_H}}{\sqrt{\alpha_L |k_j| + \alpha_H}} \ket{1}_R \right) \ket{0}_{R_k} \ket{0}_{R_t} \ket{0}_a \ket{\psi}
    \end{equation}

    \item \textbf{Apply controlled block encodings } $|0\rangle_R \langle 0|_R \otimes U_L$ and $|1\rangle_R \langle 1|_R \otimes U_H$, resulting in:
    \begin{equation}
    \begin{aligned}
    & \ket{j} \ket{l} \ket{k_j}_j \ket{t_l}_l \frac{\sqrt{\alpha_L k_j}}{\sqrt{\alpha_L |k_j| + \alpha_H}} \ket{0}_R \ket{0}_{R_k} \ket{0}_{R_t} \ket{0}_a \frac{L}{\alpha_L} \ket{\psi} \\
    & + \ket{j} \ket{l} \ket{k_j}_j \ket{t_l}_l \frac{\sqrt{\alpha_H}}{\sqrt{\alpha_L |k_j| + \alpha_H}} \ket{1}_R \ket{0}_{R_k} \ket{0}_{R_t} \ket{0}_a \frac{H}{\alpha_H} \ket{\psi} + \ket{\perp_a}
    \end{aligned}
    \end{equation}

    \item \textbf{Apply } $c\text{-}R^\dagger$:
    \begin{equation}
    \ket{j} \ket{l} \ket{k_j}_j \ket{t_l}_l \ket{0}_R \ket{0}_{R_k} \ket{0}_{R_t} \ket{0}_a \frac{k_j L + H}{\alpha_L |k_j| + \alpha_H} \ket{\psi} + \ket{\perp_{R,a}}
    \end{equation}

    \item \textbf{Apply controlled rotations on } $R_k$ \textbf{ and } $R_t$:
    \begin{equation}
    \ket{k} \ket{0} \rightarrow \ket{k} \left( \frac{\alpha_L |k| + \alpha_H}{\alpha_L K + \alpha_H} \ket{0} + \sqrt{1 - \left| \frac{\alpha_L |k| + \alpha_H}{\alpha_L K + \alpha_H} \right|^2} \ket{1} \right)
    \end{equation}
    \begin{equation}
    \ket{t} \ket{0} \rightarrow \ket{t} \left( \frac{t}{T} \ket{0} + \sqrt{1 - \left| \frac{t}{T} \right|^2} \ket{1} \right)
    \end{equation}

    Resulting in:
    \begin{equation}
    \ket{j} \ket{l} \ket{k_j}_j \ket{t_l}_l \ket{0}_R \ket{0}_{R_k} \ket{0}_{R_t} \ket{0}_a \frac{t_l (k_j L + H)}{T(\alpha_L K + \alpha_H)} \ket{\psi} + \ket{\perp_{R, R_k, R_t, a}}
    \end{equation}

    \item \textbf{Uncompute registers } $\ket{j}$ \textbf{ and } $\ket{l}$ \textbf{ using } $O_k^\dagger$ \textbf{ and } $O_t^\dagger$:
    \begin{equation}
    \ket{j} \ket{l} \ket{0}_j \ket{0}_l \ket{0}_R \ket{0}_{R_k} \ket{0}_{R_t} \ket{0}_a \frac{t_l (k_j L + H)}{T(\alpha_L K + \alpha_H)} \ket{\psi} + \ket{\perp_{R, R_k, R_t, a}}
    \end{equation}

    \item \textbf{Define full controlled block encoding operator } $U_{t(kL+H)}$:
    \begin{equation}
    (\bra{0}_{a'} \otimes I) U_{t(kL+H)} (\ket{0}_{a'} \otimes I) = \sum_{j=0}^{M_k-1} \sum_{l=0}^{M_t-1} \ket{j} \bra{j} \otimes \ket{l} \bra{l} \otimes \frac{t_l (k_j L + H)}{T(\alpha_L K + \alpha_H)}
    \end{equation}

    \item \textbf{Use this as input to the QSVT circuit for } $e^{-iT(\alpha_L K + \alpha_H)\tilde{H}}$, where:
    \begin{equation}
    \tilde{H} = \frac{t}{T} \cdot \frac{kL + H}{\alpha_L K + \alpha_H}
    \end{equation}

    \item \textbf{Construct the select oracle } SEL:
    \begin{equation}
    \text{SEL} = \sum_{j=0}^{M_k - 1} \sum_{l=0}^{M_t - 1} \ket{j} \bra{j} \otimes \ket{l} \bra{l} \otimes W_{j,l}, \quad \text{where } W_{j,l} \text{ block encodes } V_{j,l} \approx U(k_j, t_l)
    \end{equation}

    \item \textbf{Final approximation of } $h(A)$ \textbf{ via LCU:}
    \begin{equation}
    (O_{c,l}^\dagger \otimes O_{\hat{c},l}^\dagger \otimes I)\text{SEL}(O_{c,r} \otimes O_{\hat{c},r} \otimes I) = \frac{1}{\|c\|_1 \|\hat{c}\|_1} \sum_{j,l} c_j \hat{c}_l V_{j,l} \approx \frac{1}{\|c\|_1 \|\hat{c}\|_1} h(A)
    \end{equation}
\end{enumerate}

\section{Quantum algorithm for the Laplace transform}
In this section, we present our novel quantum algorithm to perform the QLT. This algorithm is built upon the Lap-LCHS framework introduced in the previous section, but with crucial modifications and optimizations specifically tailored for the QLT. 

\subsection{Quantum Laplace Transform}
\label{sec:novel_ql_algo}

The detailed steps for implementing QLT via Lap-LCHS framework are given below:

\begin{enumerate}
    \item \textbf{Initialization:} We consider a diagonal matrix A whose diagonal elements encode the discrete Laplace variable values $s_0, s_1, s_2, \ldots$. To achieve the speedup, we assume that the real components ($a_j \geq 0$) and imaginary components ($b_j$) of $s_j = a_j + i b_j$ are equal ($a_j=b_j$) and form the same arithmetic progression that is the same first term and common difference (e.g., $s_j = 1+1j$, $2+2j$, $3+3j$, etc.). As rigorously demonstrated later, this condition ensures the superpolynomial speedup. The conclusion extends straightforwardly to the more general case where $a_j \ne b_j$, provided that both sequences $\{a_j\}$ and $\{b_j\}$ independently form separate arithmetic progressions.
    \begin{equation}
    A = \begin{pmatrix}
    s_0 & 0 & 0 & 0 \\
    0 & s_1 & 0 & 0 \\
    0 & 0 & s_2 & 0 \\
    0 & 0 & 0 & s_3
    \end{pmatrix}, \text{ where } s_j = a_j + i b_j
    \label{eq:A_matrix_def}
    \end{equation}

    \item \textbf{Decomposition:} The matrix A is decomposed into its Cartesian components, $A = L + iH$, where L and H are Hermitian matrices given by:
    \begin{equation}
    \begin{split}
    L &= \frac{A+A^\dagger}{2} \\
    H &= \frac{A-A^\dagger}{2i}
    \end{split}
    \end{equation}
    Given that A is a diagonal matrix with elements $s_j = a_j + i b_j$, L and H will also be diagonal matrices with real entries:
    \begin{equation}
    L = \begin{pmatrix}
    a_0 & 0 & 0 & 0 \\
    0 & a_1 & 0 & 0 \\
    0 & 0 & a_2 & 0 \\
    0 & 0 & 0 & a_3
    \end{pmatrix}, \quad
    H = \begin{pmatrix}
    b_0 & 0 & 0 & 0 \\
    0 & b_1 & 0 & 0 \\
    0 & 0 & b_2 & 0 \\
    0 & 0 & 0 & b_3
    \end{pmatrix}
    \label{eq:L_H_matrices}
    \end{equation}

    \item \textbf{Laplace Transform based Linear Combination of Hamiltonian Simulation:}
    Now we use Lap-LCHS to do eigenvalue transformation of our diagonal matrix $A$.
    \begin{equation}
            \mathcal{}h(A) = \int_0^\infty g(t) e^{-At} \, dt = \int_0^\infty \int_{\mathbb{R}} \frac{f(k) g(t)}{1 - i k} e^{-i t (k L + H)} \, dk \, dt
    \end{equation}
    where $f(k) = \frac{1}{2\pi e^{-2^\beta} e^{(1 + i k)^\beta}}, \quad \beta \in (0, 1)$, is an asymptotically near-optimal kernel. Discretize this as:
    \begin{equation}
    \mathcal{}h(A) \approx \sum_{j=0}^{M-1} \sum_{\ell=0}^{M-1} c_j \hat{c}_\ell U(k_j, t_\ell), \quad U(k_j, t_\ell) = e^{-i t_\ell (k_j L + H)}
    \end{equation}
    $c_j$ and $\hat{c}_\ell$ are same as defined in previous Section \ref{Lap-LCHS}. $g(t)$ denotes the continuous-time function to be Laplace-transformed, while $\hat{c_l}$ represents its discretized form.

    \item \textbf{Quantum Registers:} Our algorithm employs three quantum registers, significantly reducing the quantum circuit width compared to standard Lap-LCHS by removing unnecessary auxiliary registers (we are removing $k_j,t_l$ and all other unnecessary registers except $j$, $l$, $|\psi\rangle$):
    \begin{itemize}
        \item $|\psi\rangle$: The initial input quantum state (an equal superposition state $H^{\otimes n}|0\rangle$).
        \item $|0\rangle_j$: An ancilla register for indexing over $j$ values.
        \item $|0\rangle_l$: An ancilla register for indexing over $l$ values.
    \end{itemize}
    The initial state is thus $|0\rangle_j |0\rangle_l |\psi\rangle$. 
    
    \item \textbf{Preparation: }The step involves preparing the ancilla registers with coefficients $\sqrt{c_j}$ and $\sqrt{\hat{c_l}}$ using efficient oracles $U_j$ and $U_l$ (assuming such efficient state preparation exist). $\hat{c_l}$ encodes discretized $g(t)$ at $M_t$ points and similarly $c_j$ encodes the discretized $f(k)$ at $M_k$ points. This step corresponds to the preparation step of LCU:
    \begin{equation}
    \ket{0}_j \ket{0}_\ell \ket{\psi} 
    \xrightarrow{U_{j},\,U_{\ell}} 
    \sum_{j=0}^{M_k-1} \sum_{\ell=0}^{M_t-1} \sqrt{c_j}  \sqrt{\hat{c_l}}\ket{j} \ket{\ell} \ket{\psi}/||c_j||_1 ||\hat{c_l}||_1
    \end{equation}

   \item \textbf{Apply SELECT :} This is the core computational step which is responsible for speedup in our algorithm. The SELECT operator acts as:
    \begin{equation}
    \sum_{j=0}^{M_k-1} \sum_{\ell=0}^{M_t-1} \sqrt{c_j}  \sqrt{\hat{c_l}}\ket{j} \ket{\ell} \ket{\psi}/||c_j|| ||\hat{c_l}|| \rightarrow  \sum_{j=0}^{M_k-1} \sum_{\ell=0}^{M_t-1} \sqrt{c_j}  \sqrt{\hat{c_l}} |j\rangle |\ell\rangle  U(k_j, t_\ell) |\psi\rangle/||c_j||_1 ||\hat{c_l}||_1 
    \end{equation}

    The exact form of SELECT operator is:
\begin{equation}
    \begin{split}
\text{SEL} = \prod_{a=0}^{d-1} \prod_{b=0}^{d'-1} \bigg[& |0\rangle\langle 0| \otimes |0\rangle\langle 0| \otimes I_n \\
& \quad + |0\rangle\langle 0| \otimes |1\rangle\langle 1| \otimes e^{-i h_t (-K+1)L (2^b)} \\
& \quad + |1\rangle\langle 1| \otimes |0\rangle\langle 0| \otimes I_n \\
& \quad + |1\rangle\langle 1| \otimes |1\rangle\langle 1| \\
& \qquad \otimes \begin{multlined}[t]
 (e^{-i h_t h_j L (2^{a+b})}) \\ \cdot(e^{-i h_t (-K+1)L 2^b}) \bigg]
\end{multlined}
\end{split}
\end{equation}
    Here $t_l = l h_t$, $k_j = -K + j h_j$ and $a,b$ are the indexing of each bit when $l$ and $j$ are represented in binary form. The detailed structure of the SELECT operator, particularly for the one-qubit QLT case, is presented in Sec. \ref{sec:deriving_select_form} and can be extended to multi-qubit cases.  The total number of controlled unitaries needed to be implemented are decreased from $ \mathcal {O}(N^2)$ to $\mathcal{O}\!\bigl(\log^{2} N\bigr)$ ($N$ is size of our individual ancilla registers that we are using for $j$ and $l$). We have done this by converting our discretized summation into a product of unitaries by using a clever trick (uses binary representations) from QPE \cite{lin2022lecturenotesquantumalgorithms}, which reduces the number of unitaries to be implemented. These new unitaries also eliminated the need for multi-controlled gates with more than two control qubits. Our algorithm requires multicontrolled gates only with at most two control qubits. We have also implemented each such controlled unitary using a single-step Trotterization/Product formula since $[H, L] = 0$ (both are diagonal) and their entries form same arithmetic progression, equivalent to an $\mathcal{O}(log (N_5))$ total number of such controlled $R_z$ and controlled Phaseshift gates.  Hence, the total number of gates will scale asymptotically as $\mathcal{O}(log (N_1)* log (N_2) *log (N_5))$. Taking these three registers of same size N, our gates scale as $\mathcal{O}{((log (N))^3)}$.  Hence, the $SELECT$ can be implemented very efficiently because of the diagonal, commuting structure of $L$ and $H$ and also due to their entries being in arithmetic progression form.  This is the main reason for the reduced complexity of our algorithm compared to previous works. 
    
    \item \textbf{Uncomputation:} Uncompute $|j\rangle |\ell\rangle$ registers using the transpose matrix of preparation gates (it is necessary to note that these gates need to be the transpose of the original gates in PREP, not the conjugate transpose).
    This step corresponds to the UNPREP step of LCU. The final output values of QLT at different $s$ values is stored in probabilty amplitudes of last register.
    
    \item \textbf{Circuit Diagram:}
    Below are the circuits for single-qubit QLT: one is schematic and other is explicit representation. More details of the circuit can be read from next section where we explicitly derived each operator.
\end{enumerate}
    \begin{figure}[ht]
  \centering

  \includegraphics[width=0.8\textwidth]{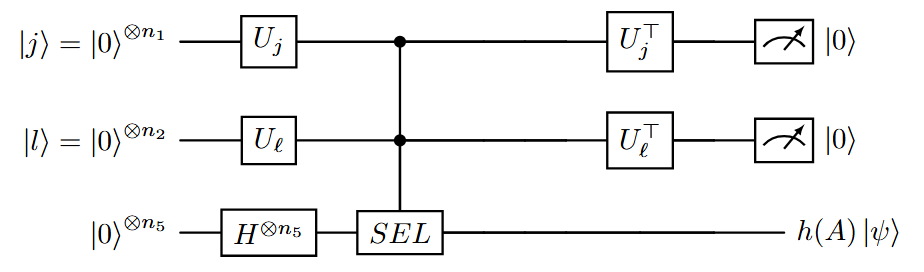}
  \caption{QLT implementation follows the standard LCU framework, consisting of the Preparation, Select, and Unpreparation operators. The target register is initialized in an equal superposition using Hadamard gates, a crucial step for our algorithm.}
  \label{fig:Circuit-for-QLT}
\end{figure}

    \begin{figure}[ht]
  \centering
  \includegraphics[width=1\textwidth]{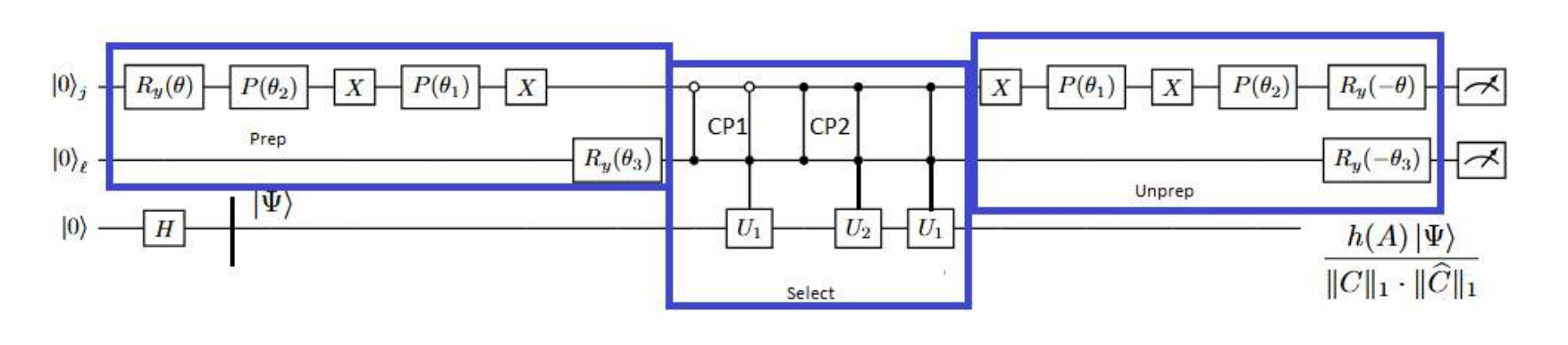}
  \caption{Explicit one-qubit realization of the QLT circuit. The diagram shows detailed decompositions of the preparation, selection, and un-preparation stages. Single-qubit rotations encode the coefficient structure of the Laplace transform, and controlled-phase as well as controlled-unitary gates implement the Select block.}

  \label{fig:one-qubit-QLT}
\end{figure}

\begin{algorithm}[H]
\caption{Quantum Laplace Transform (QLT) via Lap-LCHS (General $n$-Qubit)}
\begin{algorithmic}[1]

\Require
\Statex $n$: Number of system qubits encoding the input vector ($N = 2^n$)
\Statex $d, d'$: Number of ancilla qubits for registers $|j\rangle$ and $|l\rangle$ respectively
\Statex $|\psi\rangle$: Input quantum state on system register, $|\psi\rangle = H^{\otimes n} |0\rangle^{\otimes n}$
\Statex $C, D$: First term and common difference of the Laplace variable Arithmetic Progression ($s_j = a_j + ib_j$,
\Statex \hspace{2.4em} where $a_j = b_j$)
\Statex $K, T$: Truncation limits for the integration grid
\Statex $h_k, h_t$: Step sizes for discretization

\Ensure
\Statex $|\Psi_{\text{final}}\rangle$: State encoding the Laplace transform of the input

\Procedure{QuantumLaplaceTransform}{$|\psi\rangle, C, D, K, T, n$}

\Statex \hspace{1em} $\triangleright$ \textit{1. Initialization and Discretization Coefficients}

\State Initialize ancilla registers $|j\rangle$ and $|l\rangle$ to $|0\rangle^{\otimes d}$ and $|0\rangle^{\otimes d'}$
\State $k_j \leftarrow -K + j \cdot h_k$ for $j \in \{0, \ldots, 2^d - 1\}$
\State $t_\ell \leftarrow \ell \cdot h_t$ for $\ell \in \{0, \ldots, 2^{d'} - 1\}$
\State Calculate state-prep coefficients $c_j$ and $\hat{c}_\ell$ (encoding $f(k)$ and $g(t)$)

\Statex \hspace{1em} $\triangleright$ \textit{2. AP Pauli Decomposition for $n$ qubits (Lemma 2)}

\State $\alpha_I \leftarrow C + D\,\dfrac{2^n - 1}{2}$ \Comment{Coefficient for $I^{\otimes n}$}
\For{$m = 1$ \textbf{to} $n$}
    \State $\alpha_{Z_m} \leftarrow -D\,2^{m-2}$ \Comment{Coefficient for $Z_m$ on system qubit $m$}
\EndFor

\Statex \hspace{1em} $\triangleright$ \textit{3. LCU PREP Stage}

\State $|\Psi_0\rangle \leftarrow |0\rangle_j^{\otimes d} \otimes |0\rangle_l^{\otimes d'} \otimes |\psi\rangle$
\State $|\Psi_1\rangle \leftarrow (U_j \otimes U_l \otimes I_n)|\Psi_0\rangle$ \Comment{Encode coefficients into ancilla amplitudes}

\Statex \hspace{1em} $\triangleright$ \textit{4. LCU SELECT Stage (Optimized via Eq. 78)}

\For{$b = 0$ \textbf{to} $d' - 1$ \textbf{do}}
    \State $\theta_1 \leftarrow h_t(-K + 1)2^b$
    \State Apply C-Phase$(\theta_1 \cdot \alpha_I)$ controlled by $|l_b = 1\rangle$
    \For{$m = 1$ \textbf{to} $n$}
        \State Apply C-$R_Z(2 \cdot \theta_1 \cdot \alpha_{Z_m})$ on system qubit $m$, controlled by $|l_b = 1\rangle$
    \EndFor
    \For{$a = 0$ \textbf{to} $d - 1$ \textbf{do}}
        \State $\theta_2 \leftarrow h_t \cdot h_k \cdot 2^{a+b}$
        \State Apply CC-Phase$(\theta_2 \cdot \alpha_I)$ controlled by $|j_a = 1\rangle \otimes |l_b = 1\rangle$
        \For{$m = 1$ \textbf{to} $n$}
            \State Apply CC-$R_Z(2 \cdot \theta_2 \cdot \alpha_{Z_m})$ on system qubit $m$,
            \Statex \hspace{8em} controlled by $|j_a = 1\rangle \otimes |l_b = 1\rangle$
        \EndFor
    \EndFor
\EndFor

\Statex \hspace{1em} $\triangleright$ \textit{5. LCU UNPREP Stage}

\State $|\Psi_{\text{final}}\rangle \leftarrow (U_j^T \otimes U_l^T \otimes I_n)\,|\Psi_1\rangle$ \Comment{Apply transpose of preparation oracles}

\Statex \hspace{1em} $\triangleright$ \textit{6. Post-Selection}

\State Measure ancilla registers $|j\rangle$ and $|l\rangle$
\State Post-select on outcome $|0\rangle^{\otimes d} \otimes |0\rangle^{\otimes d'}$
\State \textbf{If} success: $|\Psi_{\text{out}}\rangle \leftarrow h(A)|\psi\rangle / (\|c\|_1 \cdot \|\hat{c}\|_1)$
\State \textbf{Else}: repeat from Step 10, or apply oblivious amplitude amplification

\State \Return $|\Psi_{\text{final}}\rangle$
\EndProcedure

\end{algorithmic}
\end{algorithm}

\section{Complexity analysis}
\label{sec:analysis}
This section rigorously analyzes the gate complexity of the QLT algorithm presented in Sec.~\ref{sec:novel_ql_algo}. We also derive the scaling of controlled unitaries, an efficient method for implementing them, and then finally the form of the overall SELECT operator, these are the core parts of our efficiency improvement. We also give the quantum circuit width complexity here.

\subsection{Circuit width complexity}
For a Laplace transform corresponding to an $N \times N$ QLT matrix, where $N = 2^{n}$ and $n$ is the number of qubits, the circuit width complexity is $\mathcal{O}(\log N)$ since all three registers individually scales as $\log N$.

\subsection{Gate Complexity}
\label{sec:gate_time_complexity}
Table \ref{tab:complexity_comparison} provides a comparative overview of the gate/time complexity for classical and quantum algorithms, specifically highlighting the benefits for Laplace transform computations when leveraging quantum computation, the structure of diagonal matrices, and the arithmetic progression structure of Laplace variable $s$.

\begin{table}[htbp]
    \centering
    \caption{Comparison of complexities for classical and quantum algorithms for QLT computations}
    \label{tab:complexity_comparison}
    \begin{tabularx}{\textwidth}{lXcc}
        \toprule
        \textbf{Matrix} & \textbf{Entry Structure} & \textbf{Classical Algo. Complexity} & \textbf{Quantum Algo. Complexity} \\
        \midrule
        Diagonal ($L=H$) & Arithmetic Progression (AP) & $\order{(M+N) \log (M+N)}$ \cite{Rabiner1969} & $\order{(\log N)^{3}}$ \\
        Diagonal ($L \ne H$) & No Structure (Non-AP) & \(\mathcal{O}((M + N)^{1.5})\) \cite{Loh2023FastDLT}
 & $\order{N^2}$ \\
        \bottomrule
    \end{tabularx}
\end{table}

This table illustrates the superpolynomial speedup achieved by the quantum algorithm, particularly when the structural properties are exploited. $M$ is number of output points and $N$ is number of input points, For comparison case take $M=N$ since for us number of input and output points are in equal numbers. We are not considering the state preparation cost which is also the case for QFT and that also assumes an efficent implementation exist to prepare the state. A similar assumption is made in Lap-LCHS~\cite{an2024laplacetransformbasedquantum}, which relies on the existence of an efficient, polylogarithmic-time state-preparation algorithm in the vector size $N$. This is justified by the Grover–Rudolph method~\cite{grover2002creatingsuperpositionscorrespondefficiently}, which constructs quantum states efficiently when the underlying function from which the input is sampled is integrable.
 Our complexity here is only for 
$SELECT$ operator which we are referring to Quantum Algorithm complexity. Including state preparation cost when we are sampling from integrable function will not change the asymptotic gate complexity. But it will affect the complexity if we are assuming an arbitrary input. A detailed proof of these complexities is provided in the following.

\subsection{Lemma 1}
\begin{lemma}
\label{lem:1}
The total number of controlled unitaries within our \textsc{select} operator scales as $\mathcal{O}((\log N)^{2})$, with each unitary being controlled by at most two qubits at a time.

\end{lemma}

\begin{proof}
The SELECT operator can be expressed as a product form by exploiting the binary representation of the indices $j$ and $l$. For $M_{k}=2^d$ for register $j$ and $M_{t}=2^{d'}$ for register $l$, we have:

\begin{equation}
\begin{split}
\text{SEL} &= \sum_{j=0}^{2^d-1} \sum_{l=0}^{2^{d'}-1} |j\rangle\langle j| \otimes |l\rangle\langle l| \otimes e^{-i t_l (k_{j}L + H)} \\
\ &= \sum_{j=0}^{2^d-1} \sum_{l=0}^{2^{d'}-1} |j\rangle\langle j| \otimes |l\rangle\langle l| \otimes e^{-il h_t ((-K+j h_j)L + H)} \\
&= \sum_{j=0}^{2^d-1} \sum_{l=0}^{2^{d'}-1} |j\rangle\langle j| \otimes |l\rangle\langle l| \otimes e^{-i l j h_t h_j L} e^{-i l h_t (-K+1)L} \\
&= \sum_{j_{d-1},...,j_{0}=0}^{1} \sum_{l_{d'-1},...,l_{0}=0}^{1} \left(\bigotimes_{a=0}^{d-1} |j_a\rangle\langle j_a|\right) \otimes \left(\bigotimes_{b=0}^{d'-1} |l_b\rangle\langle l_b|\right) \\
& \quad \otimes \left(\prod_{a=0}^{d-1} \prod_{b=0}^{d'-1} e^{-i l_b j_a h_t h_j L (2^{a+b})}\right) \left(e^{-i l_b h_t (-K+1)L (2^b)}\right)
\end{split}
\end{equation}

Here, $L = \frac{A+A^\dagger}{2}$ and $H = \frac{A-A^\dagger}{2i}$, $t_l = l h_t$, $k_j = -K + j h_j$ \cite{an2024laplacetransformbasedquantum}. The $d$ and $d'$ represent the number of qubits in the respective registers . Also, for our structure of Laplace variables 's', the real and imaginary parts are the same and hence $L=H$. The last step represents the product form enabled by the binary decomposition of $j$ and $l$. This allows us to re-express the SELECT operator as:

\begin{equation}\label{eqn: sel_op_no_reduction}
\begin{split}
\text{SEL} = \prod_{a=0}^{d-1} \prod_{b=0}^{d'-1} \bigg[& |0\rangle\langle 0| \otimes |0\rangle\langle 0| \otimes I_n \\
& \quad + |0\rangle\langle 0| \otimes |1\rangle\langle 1| \otimes e^{-i h_t (-K+1)L (2^b)} \\
& \quad + |1\rangle\langle 1| \otimes |0\rangle\langle 0| \otimes I_n \\
& \quad + |1\rangle\langle 1| \otimes |1\rangle\langle 1| \\
& \qquad \otimes \begin{multlined}[t]
 (e^{-i h_t h_j L (2^{a+b})}) \\ \cdot(e^{-i h_t (-K+1)L 2^b}) \bigg]
\end{multlined}
\end{split}
\end{equation}
The $\prod$ symbol is used in a slightly unconventional way here, denoting a tensor product on the first two (ancilla) registers and an ordinary matrix product on the third (target) register \cite{lin2022lecturenotesquantumalgorithms}.

This product form clearly shows that each term inside the large bracket corresponds to a controlled unitary where the controls are on at most two qubits (one from the $j$ register and one from the $l$ register). The total number of such product terms is $(d) \times (d')$. For us $d = d' = \log N$, then the number of controlled unitaries scales as $\order{(\log N)^2}$.
\end{proof}

\subsection{Lemma 2} \label{lemma: 2}

\begin{lemma}
Let $D$ be the diagonal matrix of size $2^n \times 2^n$ defined by
\begin{equation}
\label{eq:D-def}
D \;=\; \sum_{x=0}^{2^n-1} (a + d x)\, |x\rangle\langle x| ,
\end{equation}
where $n \in \mathbb{N}$, $a,d \in \mathbb{R}$ and $x$ is identified with its binary expansion
\begin{equation}
\label{eq:x-binary}
x \;=\; \sum_{j=1}^n 2^{j-1} x_j, 
\qquad x_j \in \{0,1\}.
\end{equation}
Consider the Pauli expansion
\begin{equation}
\label{eq:Pauli-expansion}
D \;=\; \sum_{P \in \mathcal{P}_n} \alpha_P P,
\end{equation}
where $\mathcal{P}_n$ is the $n$-qubit Pauli group and
\begin{equation}
\label{eq:alpha-def}
\alpha_P 
\;=\; \frac{1}{2^n} \sum_{x=0}^{2^n-1} (a + d x) \langle x | P | x \rangle.
\end{equation}
Then the only Pauli strings $P$ with nonzero coefficients $\alpha_P$ are 
the identity $I^{\otimes n}$ and the single-qubit $Z$ strings
\begin{equation}
\label{eq:Zj-def}
Z_j \;\equiv\; I^{\otimes (j-1)} \otimes Z \otimes I^{\otimes (n-j)},
\qquad 1 \leq j \leq n.
\end{equation}
\end{lemma}

\begin{proof}
First, observe that if $P$ contains any factor $X$ or $Y$, then
\begin{equation}
\label{eq:XY-diagonal-zero}
\langle x | P | x \rangle = 0
\quad \text{for all computational basis states } |x\rangle,
\end{equation}
because $X$ and $Y$ flip the computational basis states. Hence, from
\eqref{eq:alpha-def} and \eqref{eq:XY-diagonal-zero} we get
\begin{equation}
\label{eq:alpha-zero-XY}
\alpha_P = 0
\quad \text{if $P$ contains any $X$ or $Y$ factor}.
\end{equation}
Therefore, it suffices to consider Pauli strings $P$ that are tensor
products of only $I$ and $Z$ operators.

Let $Z_P \subseteq \{1,\dots,n\}$ denote the set of positions where $P$
has a $Z$, and define
\begin{equation}
\label{eq:r-def}
r := |Z_P|.
\end{equation}
For $x = (x_1,\dots,x_n)$ with integer value given by \eqref{eq:x-binary},
we have
\begin{equation}
\label{eq:sx-def}
\langle x | P | x \rangle 
= (-1)^{\sum_{j \in Z_P} x_j}
=: s_x.
\end{equation}
Substituting \eqref{eq:sx-def} into \eqref{eq:alpha-def}, we obtain
\begin{equation}
\label{eq:alpha-as-sum}
\alpha_P 
= \frac{1}{2^n} \sum_{x=0}^{2^n-1} (a + d x) s_x.
\end{equation}

We now interpret \eqref{eq:alpha-as-sum} as an expectation value over
a product distribution. Under the uniform measure on 
$\{0,1\}^n$, each bit $x_k$ is i.i.d.\ Bernoulli$(1/2)$. Using the
binary expansion \eqref{eq:x-binary}, we rewrite $x$ as
\begin{equation}
\label{eq:x-ckxk}
x 
= \sum_{k=1}^n c_k x_k,
\qquad 
c_k := 2^{k-1}.
\end{equation}
Thus, \eqref{eq:alpha-as-sum} becomes
\begin{equation}
\label{eq:alpha-expectation}
\alpha_P 
= \mathbb{E}\!\left[\bigl(a + d x\bigr) s_x\right]
= \mathbb{E}\!\left[\left(a + d \sum_{k=1}^n c_k x_k\right) s_x\right],
\end{equation}
where the expectation is over the random vector 
$(x_1,\dots,x_n)$ with i.i.d.\ Bernoulli$(1/2)$ components.

By linearity of expectation, \eqref{eq:alpha-expectation} can be expanded as
\begin{equation}
\label{eq:alpha-expanded}
\alpha_P
= a\,\mathbb{E}[s_x]
+ d \sum_{k=1}^n c_k\,\mathbb{E}[x_k s_x].
\end{equation}

\subsubsection*{Computation of $\mathbb{E}[s_x]$}

From \eqref{eq:sx-def}, $s_x$ depends only on the bits $\{x_j : j \in Z_P\}$:
\begin{equation}
\label{eq:sx-product}
s_x = \prod_{j \in Z_P} (-1)^{x_j}.
\end{equation}
Since the bits $\{x_j\}$ are independent and each $x_j$ is Bernoulli$(1/2)$,
we have
\begin{equation}
\label{eq:Exjminus}
\mathbb{E}\bigl[(-1)^{x_j}\bigr]
= \frac{1}{2} \bigl( (-1)^0 + (-1)^1 \bigr)
= \frac{1}{2}(1 - 1)
= 0.
\end{equation}
Therefore, for $r = |Z_P|$ defined in \eqref{eq:r-def},
\eqref{eq:sx-product} and \eqref{eq:Exjminus} give
\begin{equation}
\label{eq:Esx-cases}
\mathbb{E}[s_x]
= \prod_{j \in Z_P} \mathbb{E}\bigl[(-1)^{x_j}\bigr]
= 
\begin{cases}
1, & r = 0, \\[4pt]
0, & r \geq 1,
\end{cases}
\end{equation}
where the case $r=0$ corresponds to the empty product, which equals $1$.

\subsubsection*{Computation of $\mathbb{E}[x_k s_x]$}

We next compute $\mathbb{E}[x_k s_x]$ for $1 \leq k \leq n$, distinguishing
two cases depending on whether $k \in Z_P$ or $k \notin Z_P$.

\paragraph{Case 1: $k \notin Z_P$.}

Here $x_k$ is independent of $s_x$ (because $s_x$ depends only on 
$\{x_j : j \in Z_P\}$). Thus,
\begin{equation}
\label{eq:Exksx-factor}
\mathbb{E}[x_k s_x]
= \mathbb{E}[x_k] \, \mathbb{E}[s_x].
\end{equation}
Since $x_k$ is Bernoulli$(1/2)$, we have
\begin{equation}
\label{eq:Exk-half}
\mathbb{E}[x_k]
= \frac{1}{2}.
\end{equation}
Combining \eqref{eq:Esx-cases}, \eqref{eq:Exksx-factor}, and 
\eqref{eq:Exk-half} yields
\begin{equation}
\label{eq:Exksx-k-notin-Z}
\mathbb{E}[x_k s_x]
= \mathbb{E}[x_k] \mathbb{E}[s_x] =
\begin{cases}
\frac{1}{2}, & r = 0, \\[4pt]
0, & r \geq 1,
\end{cases}
\quad \text{for } k \notin Z_P.
\end{equation}

\paragraph{Case 2: $k \in Z_P$.}

Write
\begin{equation}
\label{eq:sx-split}
s_x = (-1)^{x_k} s_x',
\end{equation}
where
\begin{equation}
\label{eq:sx-prime-def}
s_x' := \prod_{j \in Z_P \setminus \{k\}} (-1)^{x_j}
\end{equation}
depends only on the bits $\{x_j : j \in Z_P \setminus \{k\}\}$ and is
independent of $x_k$. Using \eqref{eq:sx-split}, we have
\begin{equation}
\label{eq:Exksx-split}
\mathbb{E}[x_k s_x]
= \mathbb{E}\bigl[x_k (-1)^{x_k} s_x'\bigr]
= \mathbb{E}\bigl[x_k (-1)^{x_k}\bigr] \, \mathbb{E}[s_x'].
\end{equation}
Analogously to \eqref{eq:Esx-cases}, but now with $r-1$ factors, we get
\begin{equation}
\label{eq:Esxprime-cases}
\mathbb{E}[s_x']
=
\begin{cases}
1, & r-1 = 0 \text{ (i.e.\ } r = 1), \\[4pt]
0, & r-1 \geq 1 \text{ (i.e.\ } r \geq 2).
\end{cases}
\end{equation}
It remains to compute $\mathbb{E}\bigl[x_k (-1)^{x_k}\bigr]$. Since $x_k$ is 
Bernoulli$(1/2)$,
\begin{equation}
\label{eq:Exkminusxk}
\mathbb{E}\bigl[x_k (-1)^{x_k}\bigr]
= \frac{1}{2} \bigl( 0 \cdot (-1)^0 + 1 \cdot (-1)^1 \bigr)
= \frac{1}{2}(0 - 1)
= -\frac{1}{2}.
\end{equation}
Combining \eqref{eq:Exksx-split}, \eqref{eq:Esxprime-cases}, and 
\eqref{eq:Exkminusxk}, we obtain
\begin{equation}
\label{eq:Exksx-k-in-Z}
\mathbb{E}[x_k s_x]
=
\begin{cases}
-\dfrac{1}{2}, & r = 1, \\[6pt]
0, & r \geq 2,
\end{cases}
\quad \text{for } k \in Z_P.
\end{equation}

\subsubsection*{Summary of $\mathbb{E}[x_k s_x]$}

Equations \eqref{eq:Exksx-k-notin-Z} and \eqref{eq:Exksx-k-in-Z} together
give, for each $k$,
\begin{equation}
\label{eq:Exksx-summary}
\mathbb{E}[x_k s_x]
=
\begin{cases}
\frac{1}{2}, & r = 0, \\[4pt]
-\dfrac{1}{2}, & r = 1 \text{ and } k \in Z_P, \\[6pt]
0, & \text{otherwise}.
\end{cases}
\end{equation}

\subsubsection*{Evaluation of $\alpha_P$ in the three regimes}

We now substitute \eqref{eq:Esx-cases} and \eqref{eq:Exksx-summary} into 
\eqref{eq:alpha-expanded} and analyze the three cases $r = 0$, $r = 1$, 
and $r \geq 2$.

\paragraph{Case A: $r = 0$ (the identity).}

If $Z_P = \varnothing$, then $P = I^{\otimes n}$ and $s_x = 1$ for all $x$.
From \eqref{eq:Esx-cases} we have $\mathbb{E}[s_x] = 1$, and from
\eqref{eq:Exksx-summary} we have $\mathbb{E}[x_k s_x] = 1/2$ for all $k$.
Hence, \eqref{eq:alpha-expanded} gives
\begin{equation}
\label{eq:alpha-identity-raw}
\alpha_{I^{\otimes n}}
= a \cdot 1 + d \sum_{k=1}^n c_k \cdot \frac{1}{2}.
\end{equation}
Using $c_k = 2^{k-1}$ from \eqref{eq:x-ckxk}, we compute
\begin{equation}
\label{eq:sum-ck}
\sum_{k=1}^n c_k 
= \sum_{k=1}^n 2^{k-1}
= 2^n - 1.
\end{equation}
Substituting \eqref{eq:sum-ck} into \eqref{eq:alpha-identity-raw}, we obtain
\begin{equation}
\label{eq:alpha-identity-final}
\alpha_{I^{\otimes n}}
= a + d \cdot \frac{2^n - 1}{2}.
\end{equation}

\paragraph{Case B: $r = 1$ (single $Z$).}

Suppose $Z_P = \{j\}$ for some $1 \leq j \leq n$. Then $P = Z_j$ as in 
\eqref{eq:Zj-def}. From \eqref{eq:Esx-cases} we have $\mathbb{E}[s_x] = 0$. 
From \eqref{eq:Exksx-summary}, 
\begin{equation}
\label{eq:Exksx-r1}
\mathbb{E}[x_k s_x]
=
\begin{cases}
-\dfrac{1}{2}, & k = j, \\[6pt]
0, & k \neq j.
\end{cases}
\end{equation}
Therefore, \eqref{eq:alpha-expanded} gives
\begin{equation}
\label{eq:alpha-Zj-raw}
\alpha_{Z_j}
= a \cdot 0 + d \sum_{k=1}^n c_k\,\mathbb{E}[x_k s_x]
= d \, c_j \, \mathbb{E}[x_j s_x]
= d \cdot 2^{j-1} \cdot \left(-\frac{1}{2}\right).
\end{equation}
Hence,
\begin{equation}
\label{eq:alpha-Zj-final}
\alpha_{Z_j}
= -d \, 2^{j-2}.
\end{equation}

\paragraph{Case C: $r \geq 2$ (two or more $Z$ factors).}

If $|Z_P| = r \geq 2$, then from \eqref{eq:Esx-cases} we have 
$\mathbb{E}[s_x] = 0$, and from \eqref{eq:Exksx-summary} we have 
$\mathbb{E}[x_k s_x] = 0$ for all $k$. Substituting into 
\eqref{eq:alpha-expanded} yields
\begin{equation}
\label{eq:alpha-rge2}
\alpha_P
= a \cdot 0 + d \sum_{k=1}^n c_k \cdot 0
= 0.
\end{equation}

Combining \eqref{eq:alpha-zero-XY}, \eqref{eq:alpha-identity-final},
\eqref{eq:alpha-Zj-final}, and \eqref{eq:alpha-rge2}, we conclude:
\begin{itemize}
    \item If $P$ contains any $X$ or $Y$ factor, then $\alpha_P = 0$ by
    \eqref{eq:alpha-zero-XY}.
    \item If $P$ is a tensor product of only $I$ and $Z$ operators and
    has $r \geq 2$ $Z$ factors, then $\alpha_P = 0$ by
    \eqref{eq:alpha-rge2}.
    \item For $P = I^{\otimes n}$, the coefficient $\alpha_{I^{\otimes n}}$
    is given by \eqref{eq:alpha-identity-final}, which is generally
    nonzero.
    \item For $P = Z_j$ with $1 \leq j \leq n$, the coefficient
    $\alpha_{Z_j}$ is given by \eqref{eq:alpha-Zj-final}, which is
    also generally nonzero (unless $d=0$).
\end{itemize}
Thus, in the Pauli expansion \eqref{eq:Pauli-expansion} of $D$, the only
Pauli strings with nonzero coefficients are the identity $I^{\otimes n}$
and the single-qubit $Z$ strings $Z_j$, $1 \leq j \leq n$.
\end{proof}

\subsection{Deriving the form of SELECT operator using single-step Trotterization/Product formula}
\label{sec:deriving_select_form}
From Lemma 1, the final SELECT operator can be expanded into a product of terms, each acting on the target register, controlled by at most two ancilla qubits (one from $|j\rangle$ and one from $|l\rangle$). The terms in these products involve exponentials of L and H. Since L and H are diagonal matrices with entries forming an arithmetic progression, we can use Lemma 2 to simplify their form.

Let's consider a $2 \times 2$ diagonal matrix L (for simplicity; this extends to $2^n \times 2^n$ by Lemma 2) with entries in an arithmetic progression. It can be written as:
\begin{equation}
L = \begin{pmatrix} C & 0 \\ 0 & C+D \end{pmatrix}=H
\end{equation}
where C is the first term of the arithmetic progression and D is the common difference.
According to Lemma 2, such a matrix can be decomposed into a sum of Pauli strings consisting of only $I$ and $Z$ terms, with non-zero coefficients only for $I$ and single $Z_j$ operators. For a $2 \times 2$ matrix, this means:
\begin{equation}
L = C_1 I + C_2 Z
\end{equation}
where $I = \begin{pmatrix} 1 & 0 \\ 0 & 1 \end{pmatrix}$ and $Z = \begin{pmatrix} 1 & 0 \\ 0 & -1 \end{pmatrix}$.
The coefficients $C_1$ and $C_2$ are derived as follows (for a $2 \times 2$ matrix, $n=1$ using $lemma: 2$'s proof):
\begin{equation}
\begin{split}
C_1 &= \frac{1}{2^1} \text{Tr}(I L) = \frac{1}{2} (C + C+D) = \frac{1}{2}(2C+D) \\
C_2 &= \frac{1}{2^1} \text{Tr}(Z L) = \frac{1}{2} (C - (C+D)) = \frac{-D}{2}
\end{split}
\end{equation}
Since L and H are both diagonal matrices(from Eq.(\ref{eqn: lh_mats})), they commute ($[H,L]=0$).
\begin{corollary}
Due to this, the Pauli terms in the decomposition commute with each other and hence we can use a single-step (Trotter)product formula.
\end{corollary}
  For example, if we simplify by taking $C=1$ and $D=1$ $\left(so\,\, L = \begin{pmatrix} 1 & 0 \\ 0 & 2 \end{pmatrix}\right)$, then $C_1 = 3/2$ and $C_2 = -1/2$. The exponential $e^{-i\theta L}$ becomes $e^{-i\theta(C_1I+C_2Z)} = e^{-i\theta C_1 I} e^{-i\theta C_2 Z}$ here $\theta \in \mathbb{R}$
.The term $e^{-i\theta C_1 I}$ is a global phase, and $e^{-i\theta C_2 Z}$ is a simple $R_Z$ rotation on the relevant qubit.

Each term in the SELECT product form (from Lemma 1) is a controlled version of such an exponential which translates to:
\begin{itemize}
    \item A controlled phase gate for the $e^{-i \theta \cdot C_1 I}$ part, controlled by the qubits $j_a$ and $l_b$.
    \item A controlled $R_Z$ gate for the $e^{-i \theta \cdot C_2 Z}$ part, also controlled by $j_a$ and $l_b$.
\end{itemize}

\begin{equation}\label{eqn: sel_operator_final}
\begin{split}
\text{SEL} &= \prod_{a=0}^{d-1}\prod_{b=0}^{d'-1}\left(
\begin{aligned}
&|0\rangle\langle 0| \otimes |0\rangle\langle 0| \otimes I_n + 
|0\rangle\langle 0| \otimes |1\rangle\langle 1| \otimes 
\underbrace{e^{-ih_t(-K+1)2^{b} \frac{1}{2}(2C+D)}}_{\text{P1}} \\[2pt]
&\quad \times 
\underbrace{e^{-ih_t(-K+1)2^{b} (-D/2)Z}}_{\text{U1}} + 
|1\rangle\langle 1| \otimes |0\rangle\langle 0| \otimes I_n \\[4pt]
&+ |1\rangle\langle 1| \otimes |1\rangle\langle 1| \otimes
  \Bigl(
    \underbrace{%
      e^{-ih_t h_j (2)^{a+b} \frac{1}{2}(2C+D)}
      \,e^{-ih_t (-K+1)2^{b} \frac{1}{2}(2C+D)}%
    }_{\text{P2}}
  \Bigr) \\[2pt]
&\quad \times 
\left( \underbrace{e^{-ih_t h_j (2)^{a+b} (-D/2)Z}}_{\text{U2}} \right) \\[2pt]
&\quad \times 
\left( \underbrace{e^{-ih_t (-K+1)2^{b} (-D/2)Z}}_{\text{U1}} \right)
\end{aligned}\right)
\end{split}
\end{equation}

Here for the single qubit case that we are proving for $d=d'=1$. The complexity of implementing such controlled $R_Z$ that are $U_1$ and $U_2$  gates for single qubits is constant. The extra phases are implemented using controlled phase shift gates that is $P1$ and $P2$. The structure of the SELECT operator, derived from Lemma 1, is a product of $d \times d'$ terms. Each such term requires implementing a controlled operation where the controls are on specific qubits from the $j$ and $l$ registers, and the target is the system register.
For each of the $\order{(\log N)^2}$ controlled unitaries, implementation with a single Trotter step involves controlled $R_Z$ gates and controlled phase shift gates. The number of such gates required per term scales as $\order{\log N}$.

Thus, the total gate complexity is:
\begin{equation*}
\text{Total Gates} = \text{Number of terms} \times \text{Gates per term}
\end{equation*}
If $d = d' = \log N$ (for ancilla registers $j$ and $l$), and the target system register has $n_{sys} = \log N$ qubits, then:
\begin{align*}
\text{Total Gates} &= \order{(\log N)^2} \times \order{\log N} \\
&= \order{(\log N)^3}
\end{align*}
This confirms the stated asymptotic scaling of $\order{(\log N)^3}$ for the QLT algorithm.

\section{Implementation and simulation results}
\label{sec:results}
Here, we detail both the classical (numerically) and quantum circuit implementations of our QLT algorithm. Subsequently, we validate our approach by comparing the results obtained from both the implementations.

\subsection{Numerical Lap-LCHS implementation and its comparison with exact
analytical Laplace Transform}

Under our QLT constraints, we initially simulated the discretized Lap-LCHS integral, specifically using an arithmetic progression structure for the Laplace variable \textit{s} and by also using the diagonal matrix structure of matrix $A$. The numerical implementation was done in \verb|NumPy|. The reference used for verification is the exact analytical form of the Laplace Transform for a chosen function, obtained from standard lookup tables.\newline
In this comparison, we evaluate both the numerical and analytical forms at various points, computing the absolute difference at each point, and then obtaining the sum of the differences, to quantify the overall error. From our observations, it is evident that this sum of absolute differences rapidly converges to zero as the number of sample(discretization) points increases. Figure \ref{fig: exp_exp_sin_lt} illustrates the convergence for the input functions $e^{-0.9t}$ and, $e^{-0.9t}sin(t)$. Here, the Y-axis represents the sum of absolute errors, while the X-axis represents $log_2(M_k)$, with $M_t$ implicitly doubling at each step. We can observe this fast convergence across different input functions. From the results obtained, to obtain double-digit precision, we deduce that approximately eight qubits are required in both the $M_k$ and $M_t$ registers, which correspond to the $j$ and $l$ registers in our algorithm.

\begin{figure}[h]
    \centering
    \includegraphics[width=1.0\linewidth]{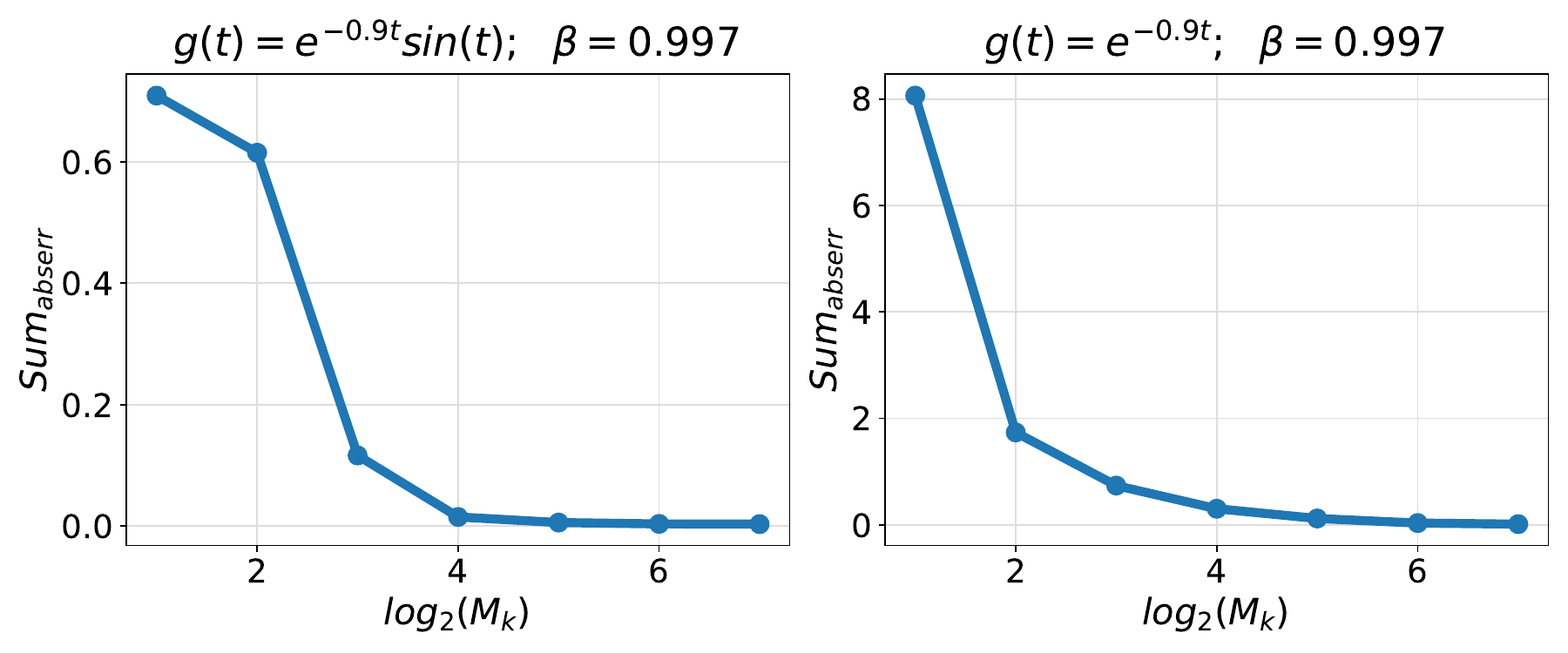}
    \caption{Comparison between numerical (classical) lap-LCHS calculation using numpy with analytical Laplace transform calculation for (a) $g(t) = e^{-0.9} sin(t)$ and (b) $g(t) = e^{-0.9t}$. }
    \label{fig: exp_exp_sin_lt}
\end{figure}

\subsection{QLT implementation on Qiskit and Pennylane}
We perform the quantum circuit implementation in three stages. First, we implement the real-weighted LCU, next, the complex-weighted LCU. Having shown the implementations of the both, we combine them both (with additional global phases) and perform the QLT circuit, for a single qubit in \verb|PennyLane v0.41.1|~\cite{bergholm2022_pennylane}.
Next, we proceed to perform a generalized QLT implementation which works for arbitrary number of qubits this is done to get desired precision.
We work with \verb|Qiskit v1.4.2|~\cite{javadiabhari2024quantumcomputingqiskit} and Pennylane for implementation of the QLT algorithm, where qiskit provides ease of access and wide reachability, whereas Pennylane offers better classical data encoding.

\subsubsection{Real weighted LCU}
We have first implemented a real-weighted LCU using two arbitrary unitary operations, $R_Y (\theta_1)$ and $R_X (\theta_2)$, with assigned weights $c_1$ and $c_2$, the SELECT operator is written as:

\begin{equation}
    SELECT = \ket{0}\bra{0} \otimes R_Y (\theta_1) + \ket{1}\bra{1} \otimes R_X(\theta_2)
\end{equation}

For the PREP part, we know the weights are real valued probability amplitudes, hence we use $R_Y (\theta)$ gates (This is a naive encoding method, please refer to Sec.17.2 (Fig 9) in \cite{Dalzell_2025} for further details).

\begin{equation}
   UNPREP-SELECT-PREP = c_1 \ket{0}\bra{0} \otimes R_Y(\theta_1) + c_2 \ket{1}\bra{1} \otimes R_X(\theta_2)
\end{equation}

The formula mentioned above is directly applicable for positive weights and trivially extensible to negative weights via a sign change. The UNPREP operation is simply the transpose of PREP, which is achieved by negating the angle of the $R_Y$ gates. 

As an example to encode two values a and b into a single qubit (assuming $|a|^2 + |b|^2 = 1$), the angle used for encoding is $\gamma = 2\,arccos\left( \frac{a}{\sqrt{a^2 +  b^2}} \right)$ and hence the $R_Y$ gate becomes the following:
\begin{equation}
    R_Y\left(\gamma\right) = \begin{bmatrix}
    cos(\frac{\gamma}{2}) & -sin(\frac{\gamma}{2}) \\ cos(\frac{\gamma}{2}) & sin(\frac{\gamma}{2})    
    \end{bmatrix} = \frac{1}{\sqrt{a^2 + b^2}} \begin{bmatrix}
        a & -b \\ b & a
    \end{bmatrix}; \gamma = 2\,arccos\left( \frac{a}{\sqrt{a^2 +  b^2}} \right)
\end{equation}

Hence, in our case, $\gamma = 2\,arccos \left( \sqrt{\frac{c_1}{\sqrt{c_1 + c_2}}}\right)$ for the $R_Y$ gate in the PREP circuit, and the angle for the $R_Y$ gate is simply $-\gamma$ as we require the transpose of the PREP circuit. The extra square root is because we encode the square root of $c_1$ and $c_2$ not just $c_1$ and $c_2$ during PREP.

The circuit diagram is shown in Figure \ref{fig:real_LCU}.
\begin{figure}[H]
    \centering
    \includegraphics[width=0.6\linewidth]{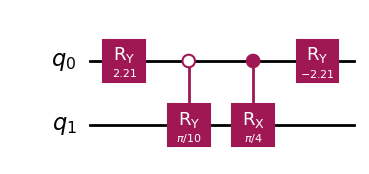}
    \caption{Real-weighted LCU implementation using single-qubit rotations and controlled operations.}
    \label{fig:real_LCU}
\end{figure}

\subsubsection{Complex weighted LCU}
We next extend our implementation to complex-weighted LCU, using weights $c_1 = 1+1i, c_2=2+2i$ for the same $R_Y(\pi/10)$ and $R_X(\pi/4)$ respectively. In the PREP stage with complex probability amplitudes, we form two complex amplitudes $z_1 = \sqrt{\alpha + i\beta}$ and $z_2 = \sqrt{\gamma + i\delta}$. These are then decomposed into their real and imaginary constituents: $z_1 = a+ib$ and $z_2 = c+id$. The rotation angles (for the PREP circuit) are then defined as:

\begin{equation}
    \begin{split}
        \theta = 2arccos\left( \left(\frac{a^2 + b^2}{a^2+b^2+c^2+d^2}\right)^{\frac{1}{4}} \right) \,\,\,\,\,\,\,\,\,\,\,\,\,\,\,\,\,\,\,\,\,\,\,\,\,\\
        \theta_1 = 2arccos\left(\sqrt{\frac{a}{\sqrt{a^2+b^2}}}\right) ;\,\,
        \theta_2 = 2arccos\left(\sqrt{\frac{c}{\sqrt{c^2+d^2}}}\right)
    \end{split}
\end{equation}

The state-preparation unitary (PREP) is implemented by the sequence: $PREP = X \cdot P(\theta_1) \cdot X \cdot P(\theta_2)\cdot R_y(\theta)$. These encoded weights include a normalization factor similar to the real-weighted case, which can be removed in classical post-processing to retrieve the correct values. The above circuit is valid when $\alpha, \beta, \gamma, \delta > 0$, and can be extended to cases where they are negative, by performing appropriate sign changes while initializing the angles. In the case of the UNPREP operation, the circuit for the same is given as $UNPREP = R_y(-\theta)\cdot P(\theta_2)\cdot X \cdot P(\theta_1) \cdot X$. Here in addition, to the angle negation of the $R_Y$ gates, the gate order is also reversed. The circuit for complex-weighted LCU is illustrated in Figure \ref{fig:complex_LCU}.

\begin{figure}[H]
    \centering
    \includegraphics[width=1.0\linewidth]{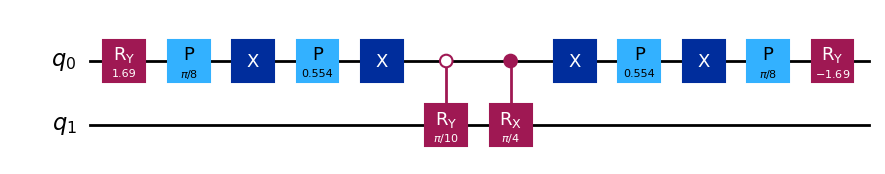}
    \caption{Quantum Circuit for Complex-weighted LCU}
    \label{fig:complex_LCU}
\end{figure}

\subsubsection{Single qubit QLT circuit}

Having performed the real-weighted LCU and complex-weighted LCU, we next tried to combine them and make a single qubit QLT circuit. The exact matrix forms of the operators required for this implementation were derived directly from the final form of the SELECT operator presented in Sec.~\ref{sec:deriving_select_form}. 

In Fig.(\ref{fig:single_qubit_qlt}), the circuit is provided for a single qubit QLT implementation for the discrete Laplace transform of the function $f(t) = e^{-0.9t}$.

\begin{sidewaysfigure}
    \centering

    \begin{minipage}{1.1\textwidth}
        \centering
        \includegraphics[width=\linewidth]{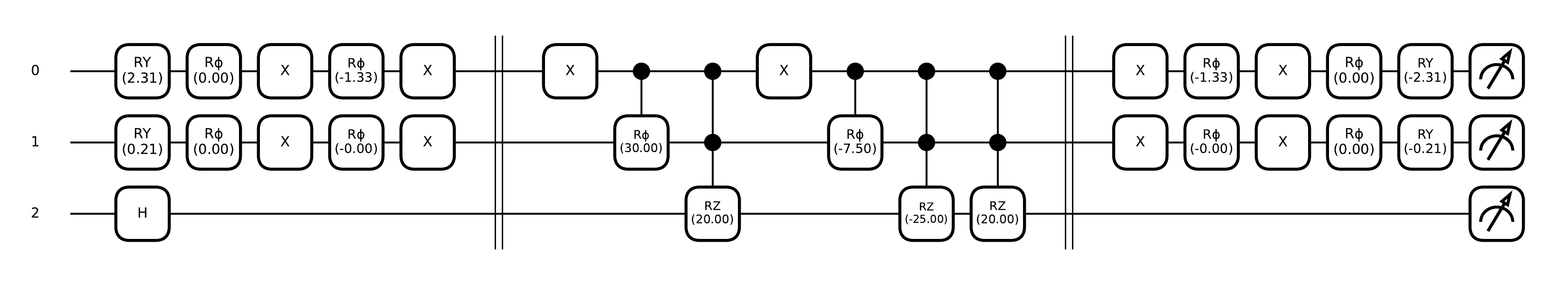}
        \captionof{figure}{Single qubit QLT circuit implemented in PennyLane.}
        \label{fig:single_qubit_qlt}
    \end{minipage}

    \vspace{1cm}  

    \begin{minipage}{1.1\textwidth}
        \centering
        \includegraphics[width=\linewidth]{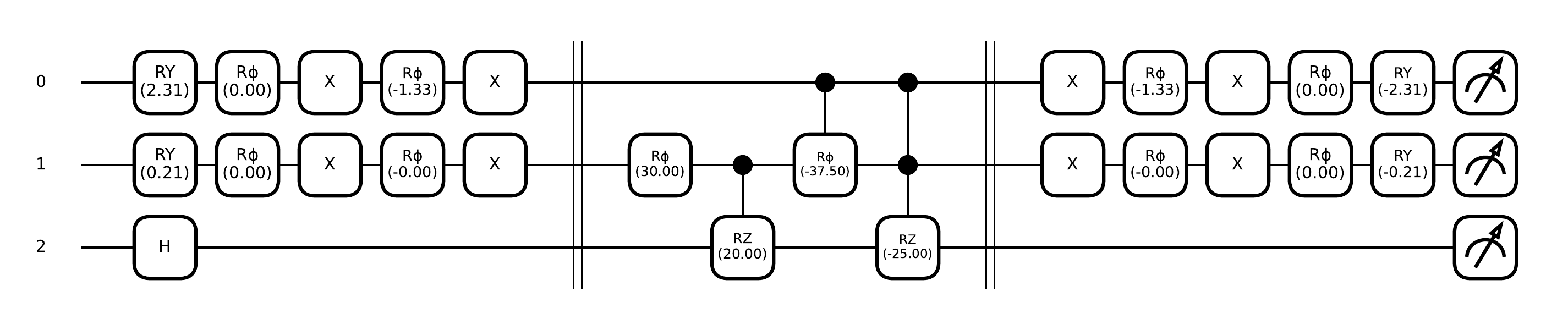}
        \captionof{figure}{Single Qubit QLT with further reduction.}
        \label{fig:further_depth_reduction}
    \end{minipage}

\end{sidewaysfigure}

\subsubsection{Quantum circuit results comparison with Classical numerical implementation}
To verify the simulation results obtained from our one-qubit QLT implementation, we compared the same with the numerical (classical) implementation of the Lap-LCHS formula. The Table~\ref{tab:qlt_lap_lchs_comparison} demonstrates that the Qiskit Implementation matches the classical implementation upto 8 digits of precision for both the probability amplitudes that are also the values for the Laplace-transformed function $(f(t) = e^{-0.9t})$ at two $s$ values $1+1i$ and $2+2i$.

\begin{table}[h!]
\centering
\caption{Comparison between QLT Implementation (Ideal State-vector Simulation in Qiskit) and Lap-LCHS Classical Numerical Implementation}
\label{tab:qlt_lap_lchs_comparison}
\begin{tabular}{ |c|c|c| } 
\hline
\textbf{Value} & \textbf{QLT Quantum Circuit} & \textbf{Lap-LCHS Classical Numerical} \\
\hline
1 & 8.96834845 - 0.9033673i & 8.96834845 - 0.9033673i \\
\hline
2 & 8.85408937 - 0.93779124i & 8.85408937 - 0.93779124i \\
\hline
\end{tabular}
\end{table}

\subsubsection{Further Reduction in gate counts in SELECT}

It is possible to further reduce the gate counts in the SELECT operator by merging certain double-controlled gates, thereby reducing the overall gate counts. The modified equation Eq. \ref{eqn: lower_depth_SEL} is the following:

\begin{equation}\label{eqn: lower_depth_SEL}
\begin{aligned}
\text{SEL} ={}& \prod_{a=0}^{d-1} \prod_{b=0}^{d'-1} \left( |0\rangle\langle 0| \otimes |0\rangle\langle 0| \otimes I_n + \ket{0}\bra{0} \otimes \ket{1}\bra{1} \otimes e^{-ih_t(-K+1)L(2^b)}\right. \\
&+ \left.\ket{1}\bra{1} \otimes \ket{0}\bra{0} \otimes I_n + \ket{1}\bra{1} \otimes \ket{1}\bra{1} \otimes \left(e^{-ih_t h_j L(2^{a+b})} \cdot e^{-ih_t(-K+1)L(2^b)}\right) \right)
\end{aligned}
\end{equation}

Although this optimization does not alter the asymptotic scaling of $\mathcal{O}((logN)^3)$, it provides meaningful practical efficiency gains. The 1 qubit QLT circuit with the reduction is shown in Fig.~\ref{fig:further_depth_reduction}.

\subsubsection{General Implementation of QLT}

Building on top of the work done previously, in this section, we describe the generalized implementation of QLT. In the previous sections, state preparation was performed using $R_y, X, P$ gates. Here, due to the exponentially increasing depth of the state-preparation circuits, we do not explicitly use the aforementioned gates, but use state-preparation operators available in Pennylane to encode the input data. This is done because state-preparation in itself is beyond the scope of this work. 

The exact circuit for QLT with 16 input sample (discretization) points is provided in Appendix 1. Hence, the input register, $M_k$ and $M_t$ that is $j$ and $l$ registers in our algorithm are 4 qubits each respectively. Here, classical numerical results and quantum implementation results converged to a common value with increased precision as predicted with an increasing number of qubits.

\section{Conclusion}
\label{sec:conclusion}
This work presents a quantum algorithm for the Laplace transform that achieves provable asymptotic improvements in gate and width complexity relative to classical approaches. By establishing the commutativity of all Pauli-decomposed exponentials involved in the construction, we enabled a single-step Trotterization method that simplifies circuit synthesis and reduces resource overhead. The complete implementation of the QLT in Pennylane, together with the supporting correctness proofs and complexity analysis is done in this work.

We also want to highlight several limitations right now to our work. Similar to the QFT, the QLT provides a subroutine speedup rather than an end-to-end algorithmic speedup; hence, it is best utilized as a primitive subroutine of a larger algorithm. Also, the preparation of arbitrary input states continues to constitute a dominant practical bottleneck.  Additionally, practical applications of QLT beyond its role in just the Laplace transformation have yet to be clearly identified, limiting its immediate applicability.

These limitations motivate several avenues for future research. Developing a quantum algorithm for the inverse discrete Laplace transform would complete the computational pipeline and address a key bottleneck in many numerical methods. Moreover, integrating QLT and its inverse with the Quantum Linear System Solvers~\cite{Harrow_2009, DUAN_hhl_2020} algorithm could enable end-to-end quantum solvers for ordinary and partial differential equations \cite{lubasch2025quantumcircuitspartialdifferential}, providing a means to assess potential quantum advantage in practical scientific settings. Finally, extending the QLT framework to applications such as quantum imaginary time evolution in the resolvent domain \cite{doi:10.1021/jp0036689} can be beneficial to simulating open quantum systems \cite{an2023quantumalgorithmlinearnonunitary}. An additional observation is that, analogous to the Lap-LCHS setting, our method extends naturally to any efficient quantum eigenvalue transformation. In particular, sampling the output at evaluation points that follow an arithmetic progression enables us to reinterpret such an eigenvalue transformation as an efficient unitary evolution on a state vector.

In summary, this work establishes QLT as a theoretically robust and computationally efficient primitive, achieving both poly-logarithmic gate and circuit width complexity in the number of qubits. While practical challenges persist, the results presented here provide a solid foundation for further refinement and application of the QLT framework in diverse scientific and engineering contexts.

\section*{Acknowledgments}

The authors express their sincere gratitude to Dr. G. Raghavan and Dr. K. Srinivasan of the Defence Institute of Advanced Technology, Pune, for their constant guidance and support throughout this work. The authors also thank Dr. Debajyoti Bera of the Indraprastha Institute of Information Technology, Delhi; Dr. Shantanav Chakraborthy of the International Institute of Information Technology, Hyderabad; Dr. Alok Shukla of Ahmedabad University, Ahmedabad and Dr. T S L Radhika of BITS Pilani, Hyderabad Campus for their timely guidance and invaluable support. The authors gratefully acknowledge the advisors of Qclairvoyance Quantum Labs for their guidance, constructive discussions, and continued support during the development of this research.

 \section*{Funding}
This research received no specific grant from any funding agency in the public, commercial, or not-for-profit sectors.

\section*{Competing interests}

The authors declare no competing interests.

\printbibliography

\appendix       

\newpage
\section*{APPENDIX 1}\label{sec : appendix_1}
\addcontentsline{toc}{section}{APPENDIX 1}

\subsection*{QLT circuit – 4 qubits}

\begin{figure}[H]
    \centering
    \includegraphics[width=0.6\linewidth]{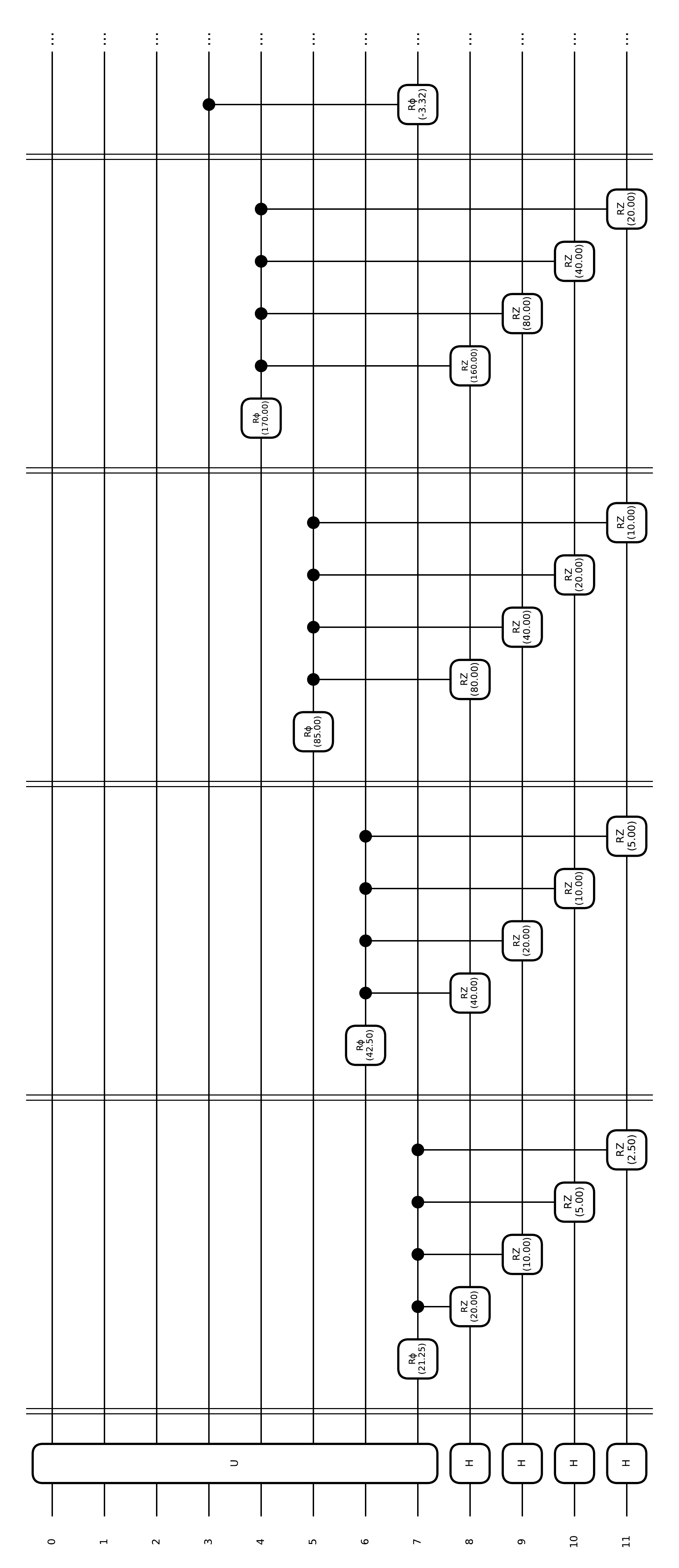}
    \caption{QLT circuit with input encoded in 4 qubits (16 points)}
    \label{fig:qlt_4_1}
\end{figure}

\begin{figure}[H]
    \centering
    \includegraphics[width=0.65\linewidth]{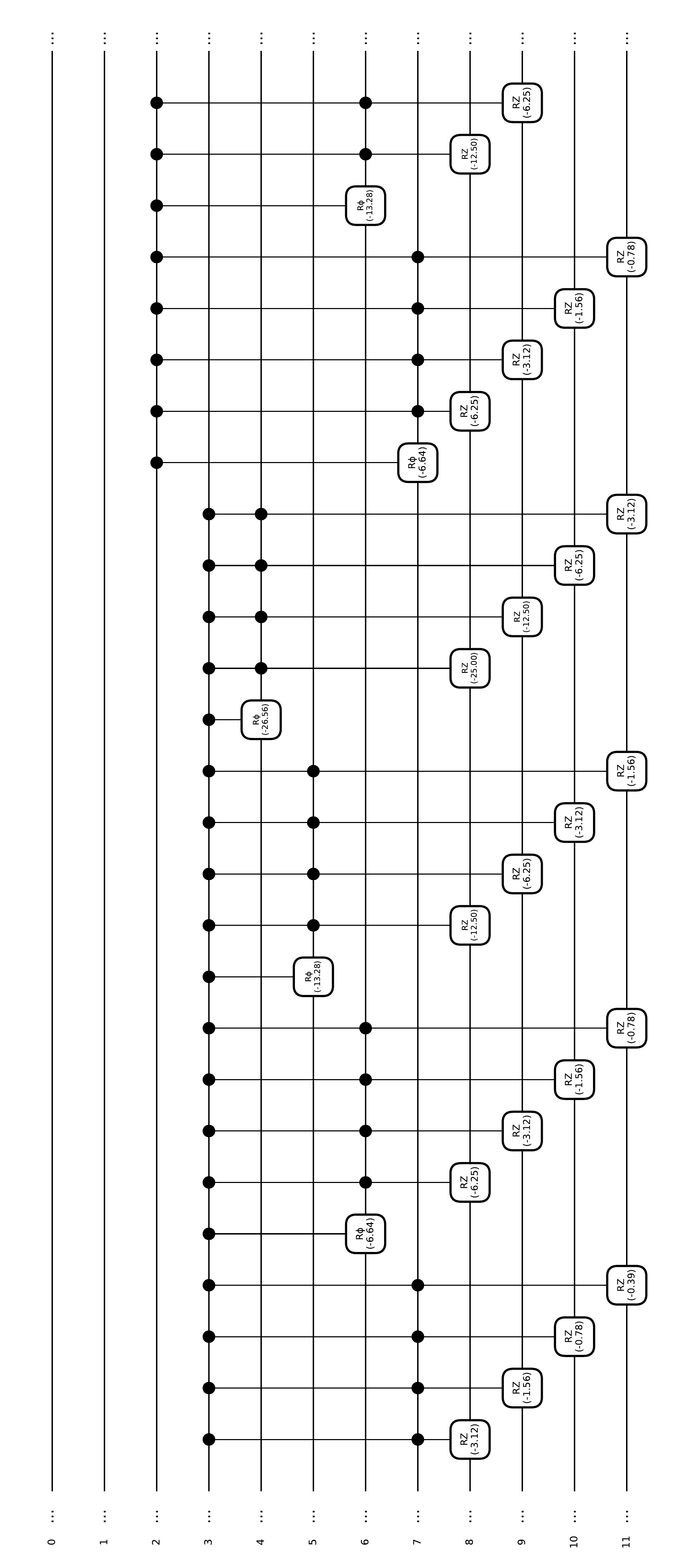}
    \caption{QLT circuit with input encoded in 4 qubits (16 points) contd.,}
    \label{fig:qlt_4_2}
\end{figure}

\begin{figure}[H]
    \centering
    \includegraphics[width=0.65\linewidth]{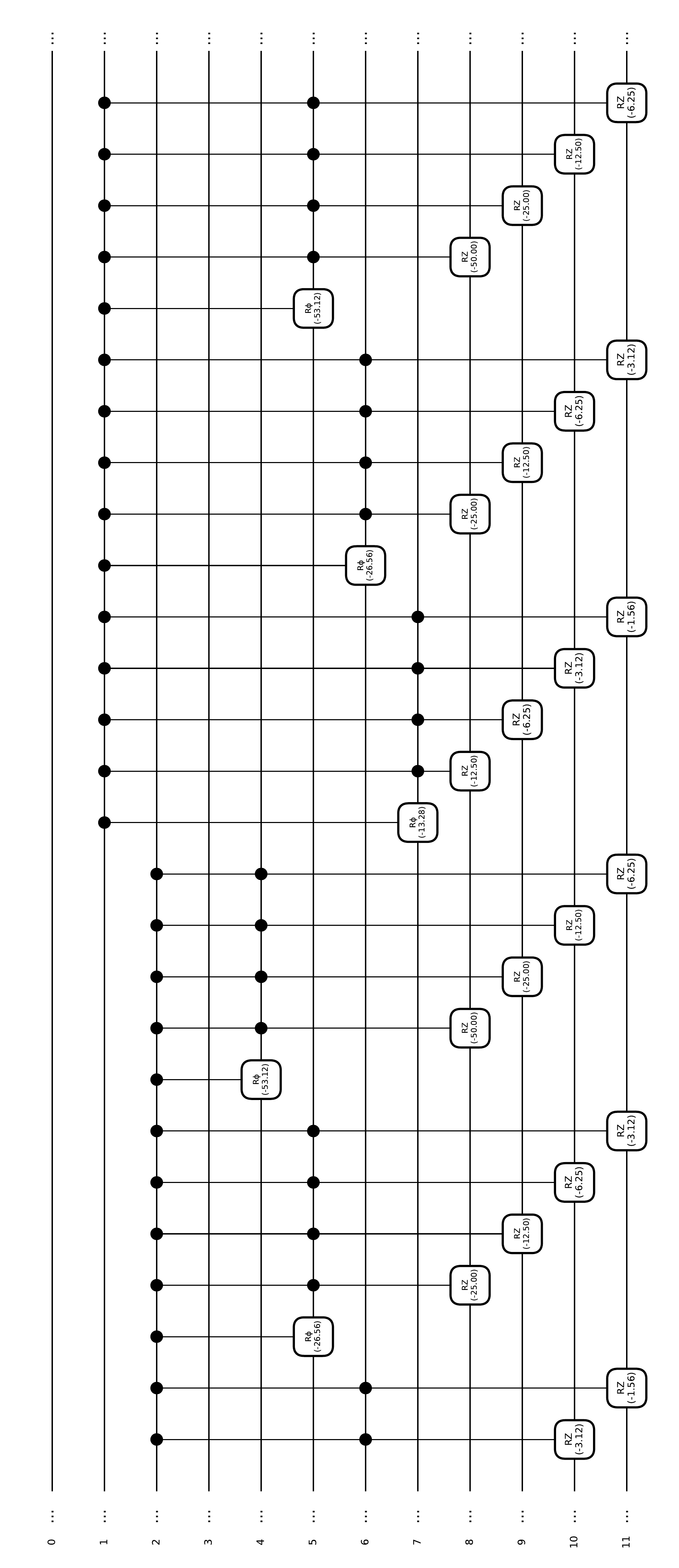}
    \caption{QLT circuit with input encoded in 4 qubits (16 points) contd.,}
    \label{fig:qlt_4_3}
\end{figure}

\begin{figure}[H]
    \centering
    \includegraphics[width=0.65\linewidth]{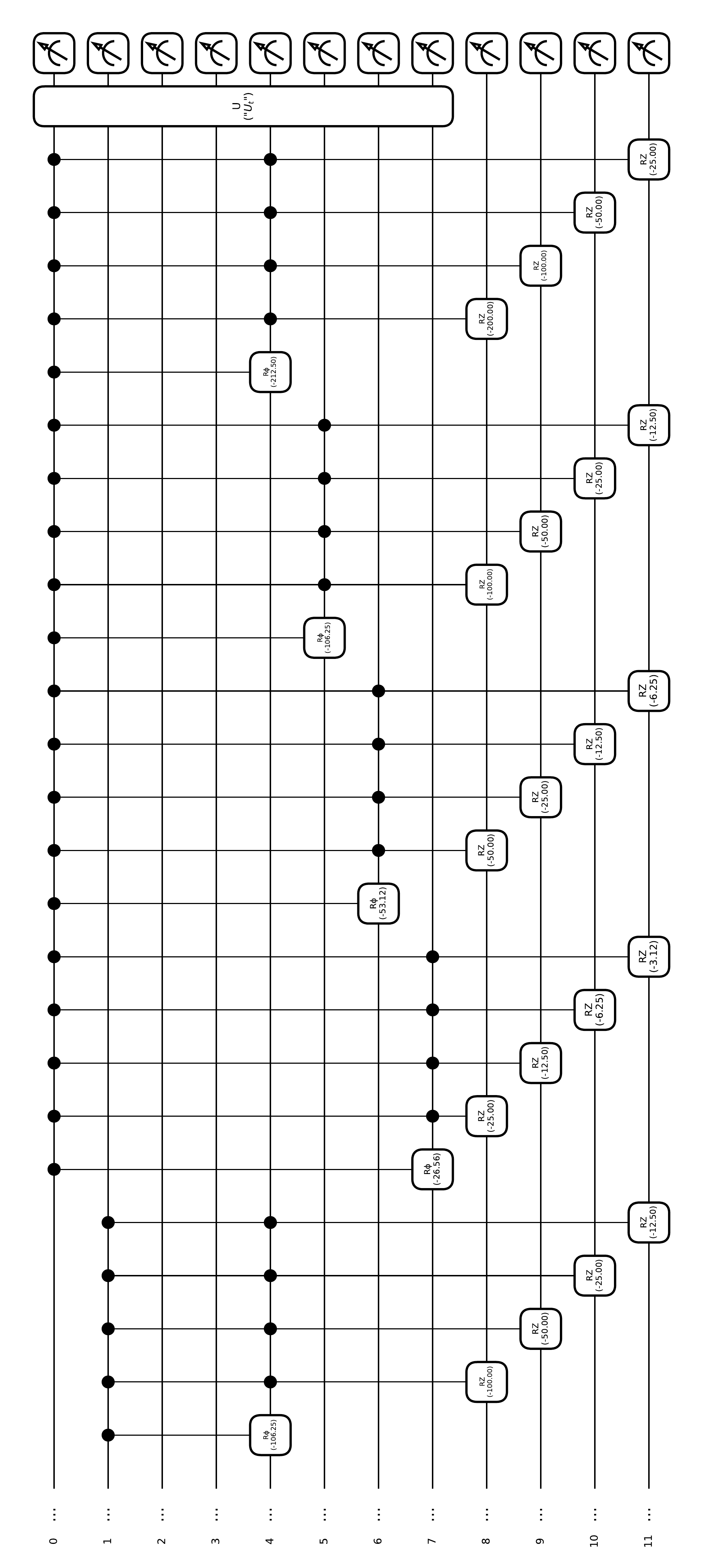}
    \caption{QLT circuit with input encoded in 4 qubits (16 points) contd.}
    \label{fig:qlt_4_4}
\end{figure}

\section*{APPENDIX 2: Notation and Symbols}
\addcontentsline{toc}{section}{APPENDIX 2: Notation and Symbols}

\begingroup
\renewcommand{\descriptionlabel}[1]{\hspace\labelsep\normalfont#1}

\begin{description}[leftmargin=2.5cm,style=multiline]
  \item[$A$] Diagonal matrix encoding the (discretized) Laplace variables: \(A=\mathrm{diag}(s_0,s_1,\dots,s_{N-1})\), where \(s_j=a_j+i b_j\).
  \item[$s_j$] The \(j\)-th Laplace parameter value, written \(s_j=a_j+i b_j\) with real part \(a_j\) and imaginary part \(b_j\).
  \item[$a_j,\; b_j$] Real and imaginary parts of \(s_j\). We assume \(\{a_j\}\) and \(\{b_j\}\) form arithmetic progressions.
  \item[$L,\; H$] Cartesian decomposition of \(A\): \(L=(A+A^\dagger)/2\) (Hermitian, real part), \(H=(A-A^\dagger)/(2i)\) (Hermitian, imaginary part).
  \item[$\mathcal{L}\{f\},\; h(A)$] Laplace transform / matrix-function notation. \(h(A)\) denotes the matrix function obtained by applying the transform to the eigenvalues of \(A\).
  \item[$g(t)$] Function on which laplace transform is done satisfying \(h(A)=\int_0^\infty g(t)\,e^{-At}\,dt\).
  \item[$f(k)$] Kernel used in the Lap-LCHS one-dimensional integral representation (\(f(k)=\bigl(2\pi e^{-2^\beta} e^{(1+ik)^\beta}\bigr)^{-1}\).
  \item[$M_k,\;M_t$] Number of discretization points used for k and t respectively.
  \item[$k_j,\; t_\ell$] Discrete grid points: \(k_j=-K+j\,h_k\) for \(j=0,\dots,M_k-1\) and \(t_\ell=\ell\,h_t\) for \(\ell=0,\dots,M_t-1\).
  \item[$h_k,\; h_t$] Grid spacings (step sizes) for the \(k\)- and \(t\)-discretizations.
  \item[$c_j,\; \hat c_\ell$] Discrete coefficients arising from quadrature weights:
    \(c_j = h_k\,f(k_j)/(1-i k_j)\) and \(\hat c_\ell = h_t\,g(t_\ell)\).
  \item[$\|c\|_1,\; \|\hat c\|_1$] \(\ell_1\)-norms (sum of absolute values) of vectors \(c=(c_j)_j\) and \(\hat c=(\hat c_\ell)_\ell\), used for normalization in PREP.
  \item[PREP, UNPREP] State-preparation and its transpose (un-preparation) used in LCU: PREP prepares the ancilla superposition proportional to \(\sqrt{c_j}\) or \(\sqrt{\hat c_\ell}\); UNPREP denotes its transpose.
  \item[LCU] Linear Combination of Unitaries: representation \(A=\sum_j \alpha_j U_j\).
  \item[SELECT (SEL)] The controlled-operation \(\sum_j |j\rangle\langle j|\otimes U_j\) (extended here to two-index select \(\sum_{j,\ell}|j\ell\rangle\langle j\ell|\otimes U(k_j,t_\ell)\)).
  \item[$U(k_j,t_\ell)$] Elementary unitary \(e^{-i t_\ell (k_j L + H)}\) appearing in the discretized Lap-LCHS sum.
  \item[Block encoding] A unitary \(U\) such that \((\langle 0^a|\otimes I)U(|0^a\rangle\otimes I)=A/\alpha\); \(\alpha\) is the normalization factor and \(a\) is the number of ancilla qubits.
  \item[QSVT / QET] Quantum Singular Value / Eigenvalue Transformation  families of techniques for applying polynomial functions to singular values / eigenvalues.
  \item[Trotterization / product formula] First-order (or higher) splitting used to approximate exponentials of sums: \(e^{-i(H_1+H_2)t}\approx(e^{-iH_1 t/r}e^{-iH_2 t/r})^r\). In this work a single-step Trotter is used because terms commute.
  \item[$\beta$] Parameter appearing in the chosen \(f(k)\) kernel, typically \(0<\beta<1\).
  \item[$N$] Problem size, typically \(N=2^n\) where \(n\) is the number of system qubits encoding the input vector.
  \item[$d,\; d'$] Number of ancilla qubits for the \(j\) and \(\ell\) (or \(k\) and \(t\)) index registers;  \(d=\log M_k,\; d'=\log M_t\).
  \item[$|j\rangle,\; |\ell\rangle$] Computational basis states of the ancilla/index registers used by PREP/SELECT.
  \item[$|\psi\rangle$] System register holding an equal superposition state.
  \item[$\mathcal{O}(\cdot)$] Asymptotic notation; $\mathcal{O}$ hides constant factors.

  \item[$\alpha_L,\; \alpha_H$] Normalization constants used when block-encoding \(L\) and \(H\) respectively.
 \item[$R_Z(\theta)$] \space Single-qubit \(Z\)-rotation gate. 
\item[$\mathrm{Phase}(\phi)$] \space Single-qubit phase-shift gate. 

  \item[Global phase] Scalar phase factor that does not affect measurement probabilities.
 
\end{description}

\endgroup

\end{document}